\documentclass[11pt]{article}

\usepackage{graphicx,epsfig,amsthm,amssymb,amsbsy,amsmath,cleveref,enumitem,setspace}
\usepackage[title]{appendix}

\newcommand{\bb}[1]{\left({#1}\right)}					% ( )
\newcommand{\sq}[1]{\left[#1\right]}						% [ ]
\newcommand{\cc}[1]{\left\{#1\right\}}					% { }
\newcommand{\op}[1]{\mathcal{#1}}
\newcommand{\ordx}[1]{{\cal O}}					% order without brackets
\newcommand{\abs}[1]{\left|#1\right|}					% | |
\newcommand{\sfrac}[2]{\mbox{$\frac{#1}{#2}$}}	
\newcommand{\hf}{\mbox{$\frac12$}}

\renewcommand{\v}[1]{{\bf #1}}

\newcommand{\BG}{$\op B$-$\op G$}
\newcommand{\tb}{\tau}
\newcommand{\subrank}{\operatorname{subrank}}
\newcommand{\Cr}[2]{\Big(\mbox{\footnotesize$\begin{array}{c}#1\\#2\end{array}$}\Big)}

\newcommand{\alphab}{{\mbox{\boldmath$\alpha$}}}

\newtheorem{theorem}{Theorem}[section]
\newtheorem{lemma}[theorem]{Lemma}
\newtheorem{corollary}[theorem]{Corollary}

\newtheorem{definition}{Definition}[section]

\crefname{equation}{}{}
\Crefname{equation}{}{}

\begin{document}

\title{Practical conditions to locate $\Sigma^1$ singularities in $\mathbb R^n$}
\title{Catastrophe conditions for $\mathbb R^n$ vector fields: \\relation to Thom-Boardman theory}
\title{When a bifurcation point is unfindable, \\look for its underlying catastrophe}
\title{Underlying catastrophes: \\how to find high codimension \\ bifurcation points of vector fields and PDEs}
\title{Elementary catastrophes underlying \\bifurcations of vector fields and PDEs}
%\title{Underlying catastrophes}
%
%
\author{Mike R. Jeffrey\footnote{\tiny Department of Engineering Mathematics, University of Bristol, Ada Lovelace Building, Bristol BS8 1TW, UK, email: mike.jeffrey@bristol.ac.uk}
}
\date{\today}

% ORCID https://orcid.org/0000-0002-3325-7211

\maketitle

\begin{abstract}
A practical method was proposed recently for finding local bifurcation points in an $n$-dimensional vector field $F$ by seeking their `{\it underlying catastrophes}'. 
Here we apply the idea to the homogeneous steady states of a partial differential equation as an example of the role that catastrophes can play in reaction diffusion. What are these `underlying' catastrophes? We then show they essentially define a restricted class of `solvable' rather than `all classifiable' singularities, by identifying degenerate zeros of a vector field $F$ without taking into account its vectorial character. As a result they are defined by a minimal set of $r$ analytic conditions that provide a practical means to solve for them, and a huge reduction from the calculations needed to classify a singularity, which we will also enumerate here. In this way, {\it underlying catastrophes} seem to allow us apply Thom's {\it elementary catastrophes} in much broader contexts. 
% in \cite{j22cat}. 
% Within this class we show that , reduce
%Here we show how these fit into the general classification of singularities of mappings provided by Thom-Boardman theory. Specifically, we show that the conditions that define {\it underlying catastrophes} are equivalent to the {\it Boardman symbols}, for corank 1 codimension $r$ singularities restricted only to a set that can be solved from a set of $r$ analytic conditions. %From this we prove, as previously conjectured, that a codimension $r$ {\it underlying catastrophe} provides a point where $r+1$ zeros of a vector field coincide. Moreover we prove that the catastrophes are reducible to local (pre-)normal forms.  
\end{abstract}

{\scriptsize
\setstretch{0.8}
\tableofcontents
\setstretch{1}
\bigskip
}

\newpage
%%%%%%%%%%%%%%%%%%%%%%%%%%%%%%%%%%%%%%%%%%%%%%%%%%%%%%
%%%%%%%%%%%%%%%%%%%%%%%%%%%%%%%%%%%%%%%%%%%%%%%%%%%%%%
\section{Introduction}\label{sec:intro}

Catastrophe theory was introduced by Ren\'e Thom to explain how discontinuities in behaviour can arise out of smooth changes in a system's equations. This marked a conceptual step in the history of calculus that was very influential in physics, particularly in optics \cite{ps96,1977zeeman}, and despite some overly speculative applications that attracted fair skepticism, its importance to science was never really in doubt, see for example the polemic \cite{catastrophe77nature} and its rebuttal \cite{berry77cat}. 

These days catastrophes have been largely absorbed into the broader theory of bifurcations, but here we will argue that a wider application of Thom's original concept is possible, using the idea of {\it underlying catastrophes} introduced in \cite{j22cat}. This concept attempts to apply the elementary catastrophes to much wider classes of systems than they were intended to encompass, basically to any multidimensional systems with general spatial and temporal variations. So it is worth summarizing Thom's catastrophes briefly, pertaining to a gradient function $\v F=\nabla V$, in preparation to apply them to any general vector field $\v F$.   

Indeed Thom's elementary catastrophes are so simple, we can encapsulate them in just \cref{fig:th}. 
\begin{figure}[h!]\centering
\includegraphics[width=0.99\textwidth]{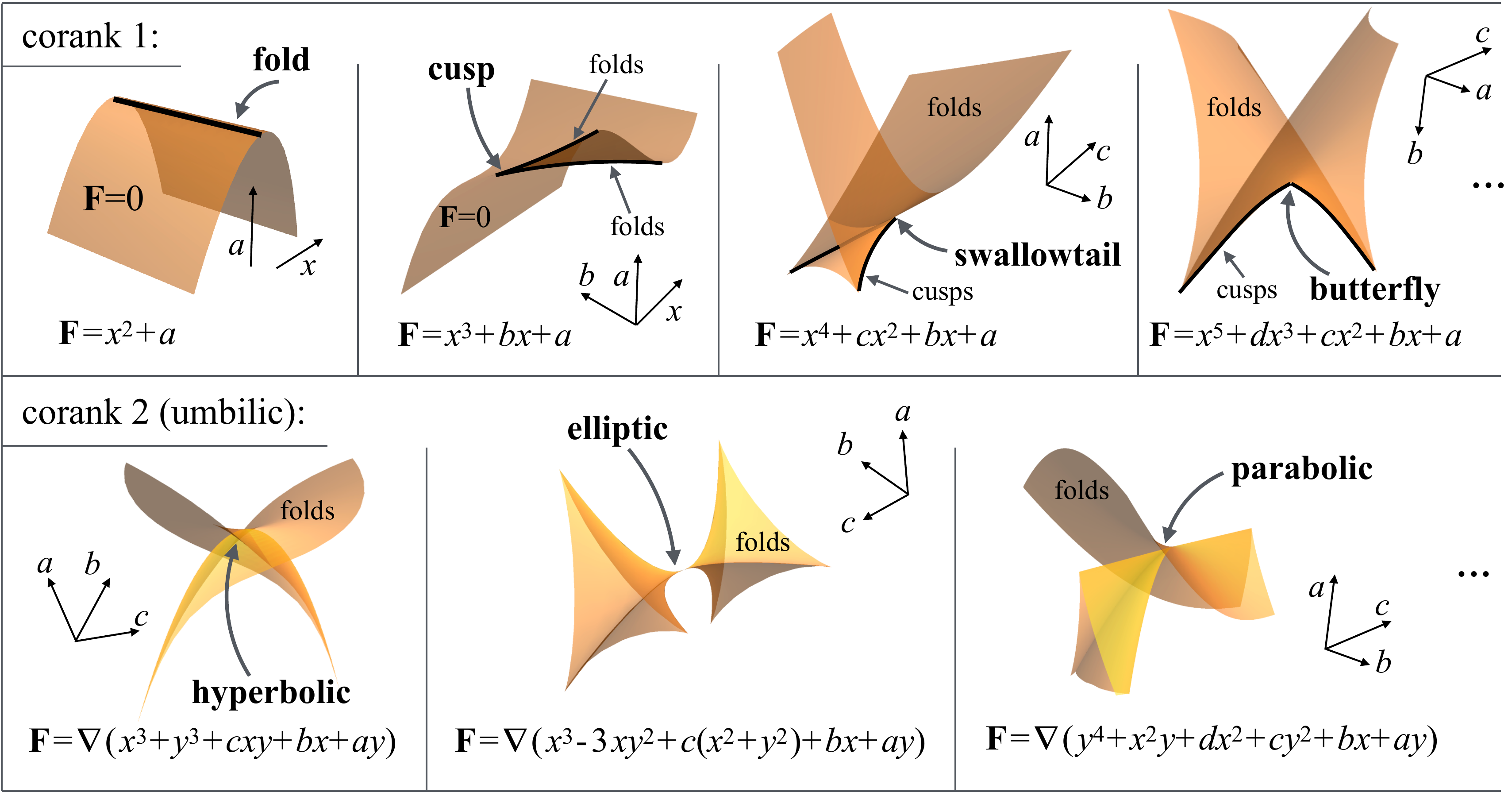}
\vspace{-0.3cm}\caption{\small\sf Thom's elementary catastrophes. A scalar function $V$ over variables $(x,y,...)$ typically has turning points where $\v F=\nabla V=0$. The catastrophe is the set of {\it folds} where these turning points `collide' as parameters $\alphab=(a,b,...)$ vary. As we add more variables $(x,y,...)$, higher `coranks' of catastrophe are possible. Note that we choose different spaces to picture these in as we add dimensions and parameters.}\label{fig:th}\end{figure}
Consider a scalar function $V:\mathbb R^n\times\mathbb R^p\to\mathbb R$, over variables $\v x\in\mathbb R^n$ and parameters $\alphab\in\mathbb R^p$. This will typically have turning points where $\nabla V=0$. Let $\v F=\nabla V$. The {\it catastrophe surface} is the manifold upon which those turning points are degenerate, colliding (as the parameters $\alphab$ vary) pairwise in folds, folds colliding in cusps, cusps colliding in swallowtails, and so on. In \cref{fig:th} we plot the catastrophe surfaces of $\v F$ based on Thom's forms in \cite{t75}, but using the specific forms laid out in the excellent exposition \cite{ps96}, in variables $\v x=(x,y,....)$ and parameters $\alphab=(a,b,c,d,...)$ (and to depict the butterfly/parabolic we set $d=0$). The simplest are the $n=1$, or `corank 1' catastrophes. In $n=2$ or more dimensions we can also have the `corank 2' {\it umbilic} catastrophes. This paper will deal in detail with the corank 1 {\it underlying catastrophes}, and the first forays into corank 2 can be found in \cite{j24cat}.  We will give explicit conditions that can be solved to find these folds, cusps, swallowtails, etc., here, but in a manner not restricted to gradient vector fields.

Notice that catastrophes provide two powerful tools: first a way to classify the topology of a function (via its stationary points), and second, simple conditions to locate changes in that topology (as stationary points appear or disappear with changes in parameters). Notice, however, that the catastrophes deal only with the turning points of a {\it potential} (scalar) function $V$, and hence the zeros of its gradient function $\v F=\nabla V$. This is what we will generalise here. We will also show how this permits the finding of singularities that are otherwise incomputable. 

Not only do bifurcations become increasingly difficult to classify as we go to systems with more variables and more parameters, but to classify any bifurcation, its location must first be known. Here lies our problem. The known classification of bifurcations (e.g. \cite{t75,a93v,k98,gh02}) involve calculations whose implicitness, and shear numerousness (as we will show here), make locating such points impossible beyond low order cases. % like those in \cite{gh02,k98}. 
Moreover it is unclear to what extent such classifications can be applied beyond mere vector fields, for example to spatiotemporal problems like reaction-diffusion. Some authors have been drawn to catastrophes as a simpler means to understand behaviour in dynamical systems and partial differential equations, and yet hindered by the fact that catastrophes apply only to scalar potentials. Recent examples include the identification of butterfly catastrophes in reaction diffusion problems in \cite{wehner20,nosov16}, a cusp catastrophe in a crowd density problem in \cite{cuspjam10}, and attempts to apply catastrophe theory to specific contexts in liquid crystals \cite{turzi09} or variational PDEs \cite{kreusser20}. The importance of nonlinear terms in reaction-diffusion has of course been clear right from the outset of the modern interest in pattern formation \cite{dawesTuring,sh77,turing52}; I will recount some of this in \cref{sec:pde}.

A more direct extension of catastrophes to such problems is made possible by a recent suggestion from \cite{j22cat}, that {\it any} bifurcation actually has an elementary catastrophe {\it underlying} it. The {\it underlying catastophe} was described in \cite{j22cat} as `a zero of a vector field encountering a singularity', and I will make this identification precise here. 
I also give an example of a reaction-diffusion problem where diffusion terms trigger an underlying catastrophe, separating a system's steady state concentrations into different mixing regimes, the boundaries between which are clearly recognisable as a butterfly catastrophe. This is done without restricting attention to scalar catastrophes at any point.

We define an underlying catastrophe as follows. 

Take a smooth vector field $\v F:\mathbb R^n\times\mathbb R^r\to\mathbb R^n$ with components $\v F=(f_1,f_2,...,f_n)$. Write $\v F=\v F(\v x,\alphab)$ for a variable $\v x\in\mathbb R^n$ and parameter $\alphab\in\mathbb R^r$, with gradient operator $\nabla=\sfrac{\partial \;}{\partial\v x}$. 
A point where the vectors $\nabla f_1,...,\nabla f_n$, are linearly dependent creates a singularity, but we can typically assume that any $n-1$ of the gradient vectors $\nabla f_1,...,\nabla f_n$, are linearly independent. The Jacobian determinant $\op B_1=\abs{\nabla\bb{f_1,,f_2,...,f_n}}$, then vanishes, and the degeneracy of the point can be characterized by a sequence of determinants $\op B_i=\abs{\nabla\bb{\op B_{i-1}, f_2, ...,  f_n}}$, letting $\op B_0=f_1$. 

The vanishing of all $\op B_i$ for $i=1,...,r$ will signal a catastrophe of codimension $r$. The system $0=\v F=\op B_1=...=\op B_r$ is solvable to find the point where a catastrophe occurs if an accompanying family of determinants $\op G_{r,k_1..k_{r-1}}$ is non-vanishing. The definition of these $\op B$ and $\op G$ determinants will be outlined below, and given more completely in \cref{sec:bg}. 

%These are much simpler, more primitive conditions than those that would be required to identify a bifurcation point of $\v F$, ignoring its vectorial (i.e. directional) character, and therefore ignoring dynamical stability. 
 %
%The $\op B_i$ and $\op G_{r,k}$ functions generalize the derivatives that define catastrophes of scalar functions. For 
This is easiest to see by considering a vector field of the form
%\begin{subequations}
\begin{align}\label{primary}
\v F=&\Big(\;f(x_1,\alphab)+\underline{\tau}\cdot\underline{x}\;,\;\lambda_2x_2\;,\;...\;,\;\lambda_nx_n\;\Big)\\
{\rm where}\quad 
&\;\;f(x_1,\alphab)%=x_1^{r+1}+\sum_{i=1}^{r}\alpha_{i}x_1^{i-1}\;,\\
	=x_1^{r+1}+\alpha_{r}x_1^{r-1}+...+\alpha_{2}x_1+\alpha_{1}\;,\nonumber
\end{align}
%\end{subequations}
with $\underline{\tau}\cdot\underline{x}=\tau_2x_2+...+\tau_nx_n$, for which 
the conditions $\op B_i=0$ reduce to just the one-dimensional conditions $\sfrac{\partial^i\;}{\partial x_1^i}f=0$ that identify elementary catastrophes of the scalar field $f$. There are three sets of parameters here, of which the $\alpha_i$ are the most important, as these unfold the catastrophe as they vary about $(\alpha_1,...,\alpha_r)=(0,...,0)$. The constants $\tau_i$ and $\lambda_i$ for $i=2,...,n$ are required to be non-zero for the equations $0=\v F=\op B_1=...=\op B_r$ to be solvable, and they determine the values of the $\op G_{r,k_1..k_{r-1}}$ to be defined below. The expression \cref{primary} is called the {\it primary form} of the catastrophe in \cite{j22cat}, and is similar to Arnold's expressions for the {\it principle family of class $A_{r+1}$} for singularities of vector fields from \cite{a93v}. The extent to which this constitutes a local model or `normal form' of a codimension $r$ underlying catastrophe, such that $\v F$ can be transformed into such an expression locally, will be discussed in a forthcoming work \cite{jc23tbprocedure}. 
%, and the $\op G_{i,k}\neq0$ reduce to their familiar non-degeneracy conditions. 

More completely, the following was proposed in \cite{j22cat}. 
  
\begin{definition}\label{def:sings}
A vector field $\v F:\mathbb R^n\times\mathbb R^r\to\mathbb R^n$ exhibits an {\bf underlying catastrophe} of codimension $r$ if it has a point $(\v x_*,\alphab_*)\in \mathbb R^n\times\mathbb R^r$ at which $\v F=0$ and
\begin{align}\label{singr}
&\op B_1=...=\op B_{r}=0\;.
\end{align} 
We say the catastrophe is {\bf full} if non-degeneracy conditions
	\begin{align}\label{singdeg}
%	\op T\neq0\quad{\rm and}\quad
	\op G_{r,k_1...k_{r-1}}\neq0\;,
	\end{align}
	hold at $(\v x_*,\alphab_*)$ for every $k_j\in\cc{1,...,n}$, $j=1,...,r-1$. The functions $\op B_i$ and $\op G_{r,k_1...k_{r-1}}$ are defined in \cref{sec:bg} \cref{def:bg} (and outlined more informally below in \cref{Br3}-\cref{Gr3}).  
	Following Thom's elementary catastrophes we call these underlying catastrophes the {\it fold} ($r=1$), {\it cusp} ($r=2$), {\it swallowtail} ($r=3$), etc.
%	The functions $\op B_i$ and $\op G_{i,k}$ are given in \cref{def:bg} below. 
\end{definition}

The functions $\op G_{r,k_1...k_{r-1}}$ are a set of extended determinants whose non-vanishing ensures non-degeneracy of the catastrophe and solvability of the conditions \cref{singr}. They are more numerous that the $\op B_i$s but it is straightforward to calculate them and just verify that they are nonzero. First, note that in defining the $\op B_i$ determinants, 
		\begin{align}\label{Br3}
		\op B_1&=\abs{\nabla\bb{ \hspace{0.04cm}f_1\hspace{0.04cm}, f_2...,  f_n}}\;,\nonumber\\
		\op B_2&=\abs{\nabla\bb{ \op B_1, f_2...,  f_n}}\;,\\  
		\op B_3&=\abs{\nabla\bb{ \op B_2, f_2...,  f_n}}\;,\;\;...\nonumber
		\end{align}
we have chosen to replace the first component of $\v F$ each time, but one could choose to replace one of the other components. Let $\op B_{2,k_1}$ be the alternative to $\op B_2$ with $\op B_1$ instead place in the $k_1^{th}$ component of $\v F$, and so iteratively let each $\op B_{i,k_1...k_{i-1}}$ be an alternative to $\op B_i$ with $\op B_{i-1,k_1...k_{i-2}}$ placed in the $k_{i-1}^{th}$ component of $\v F$, giving
%the $\op B_i$ are a sequence of Jacobian determinants in which the first row of $\v F$ is replaced by the preceding $\op B_{i-1}$, that is 
%%		The resulting functions $\op B_i$ provide the $r$ different conditions $\op B_1=...=\op B_r=0$ that must be solved in \cref{singr} to locate a codimension $r$ {underlying catastrophe} of $\v F$. 
%%
%In defining each successive $\op B_i$ 
%The existence of those alternatives are captured in the functions $\op G_{r,k}$. First let $\op B_{2,k_1}$ denote the counterpart to $\op B_2$ where $\op B_1$ is placed in the $k_1^{th}$ component of $\v F$, 
		\begin{align}
			\op B_{2,k_1}&=\abs{\frac{\partial(f_1,...,f_{k_1-1},\op B_1,f_{k_1+1},...,f_n)}{\partial(x_1,...,x_n)}}\;,	\nonumber\\
%		\end{align}
%then let $\op B_{3,k_1k_2}$ denote the counterpart to $\op B_3$ where $\op B_{2,k_1}$ is placed in the $k_2^{th}$ component, 
%		\begin{align}
			\op B_{3,k_1k_2}&=\abs{\frac{\partial(f_1,...,f_{k_2-1},\op B_{2,k_1},f_{k_2+1},...,f_n)}{\partial(x_1,...,x_n)}}\;, \quad...
		\end{align}
		and so on. Hence identifying $\op B_{2,1}\equiv\op B_2$, $\op B_{3,11}\equiv\op B_{3}$, etc. brings us back to the choices \cref{Br3}. The functions $\op G_{r_k}$ essentially establish an independence between these functions by evaluating extended determinants over $\mathbb R^n\times\mathbb R^r$ given by
		\begin{align}\label{Gr3}
			\op G_{1}&=\abs{\frac{\partial(f_1,...,f_{n},\op B_{1})}{\partial(x_1,...,x_n,\alpha_1)}}\;,\nonumber\\
			\op G_{2,k_1}&=\abs{\frac{\partial(f_1,...,f_{n},\op B_{1},\op B_{2,k_1})}{\partial(x_1,...,x_n,\alpha_1,\alpha_2)}}\;,\nonumber\\
			\op G_{3,k_1k_2}&=\abs{\frac{\partial(f_1,...,f_{n},\op B_{1},\op B_{2,k_1},\op B_{3,k_1k_2})}{\partial(x_1,...,x_n,\alpha_1,\alpha_2,\alpha_3)}}\;,\quad ...
		\end{align}
These provide the non-degeneracy conditions to be evaluated in \cref{singdeg}. Again, although these $\op G_{i,k}$ are numerous, one need only check that they are non-zero. For the primary forms \cref{primary} they extend the familiar non-degeneracy conditions for catastrophes of a scalar function $f$. 

The key to underlying catastrophes is that the conditions \cref{singr} typically give a set of $r$ algebraic equations which, along with $\v F=0$ and provided \cref{singdeg} hold, are solvable in $n$ variables and $r$ parameters. Were we to attempt to find a point in a system where all of the possible $\op B_{i,k_1...k_{i-1}}$ vanished we would have roughly $n^r$ conditions, requiring a system to have this many parameters to satisfy them in general. The complete classification of singularities and bifurcations contain many more cases, which in general are defined by more conditions still, the number of them growing superfactorially (as factorials of factorials, as we will show) in $n$ and $r$. The underlying catastrophes reduce our attention, by satisfying \cref{singdeg}, to cases with the minimal number of conditions $n+r$, and provides the simple determinant conditions \cref{singr} to find them.

We will first give a suggestion of the practical possibilities of these conditions with an example applying underlying catastrophes to a reaction-diffusion equation in \cref{sec:pde}. 
We then give more formal statements and general expressions for these `\BG' determinants in \cref{sec:bg}, recapping from \cite{j22cat}. 
To help understand how underlying catastrophes relate to established bifurcation theory, we will show that we can locate the conditions $\op B_i=0$ within the Thom-Boardman classification of singularities, treating the function $\v F$ as a general mapping rather than specifically a vector field. We outline the Thom-Boardman scheme in \cref{sec:B}, and show how the vast number of calculations involved (which we enumerate in \cref{sec:count}), reduce to the \BG~determinants of underlying catastrophes for corank 1 singularities. The proof of this result is instructive, so we explore it in some length in \cref{sec:TBBproof}, ending with some final remarks in \cref{sec:conc}.

%%%%%%%%%%%%%%%%%%%%%%%%%%%%%%%%%%%%%%%%%%%%%%%%%%%%%%
%%%%%%%%%%%%%%%%%%%%%%%%%%%%%%%%%%%%%%%%%%%%%%%%%%%%%%
\section{Example: a reaction-diffusion catastrophe}\label{sec:pde}

Let us write a reaction-diffusion equation for two species $A$ and $B$ on a spatiotemporal domain $(t,x,y)$. The usual introduction to Turing patterns starts from equations with a linear reaction term, some
\begin{align}\label{pde0}
\sfrac{\partial\;}{\partial t}A&=\nabla^2d_1A+\mu_1A+\mu_2B\;,\nonumber\\
\sfrac{\partial\;}{\partial t}B&=\nabla^2d_2B+\nu_1A+\nu_2B\;,
\end{align}
with reaction coefficients $\mu_i,\nu_i$, and diffusion constants $d_i$. Near a steady state one assumes that wave-like  solutions permit us to replace the derivative $\nabla^2$ by a factor $-k^2$. 

At the end of his seminal paper \cite{turing52}, Turing noted that the linear reaction rates in \cref{pde0} were valid for a system just leaving its homogeneous state, but more commonly one would expect to see a system passing from one pattern to another, %for that nonlinear terms were needed (he proposed to follow this up with a model of phyllotaxis which sadly was not completed, but his unfinished notes were reproduced in his collected works \cite{turingcol}). 
and the study of such non-linear reaction-diffusion has continued intensively since, but let us revisit the problem using catastrophes. Assuming we can replace $\nabla^2$ by the factor $-k^2$ (of course this does not hold exactly in a nonlinear system, but we employ it here just as a preliminary investigation of nonlinear perturbations), take an example 
\begin{align}\label{pde00}
\sfrac{\partial\;}{\partial t}A&=-k^2d_1A+\mu_1A+\mu_2B+\alpha_1B^2+\alpha_2B^3\;,\nonumber\\
\sfrac{\partial\;}{\partial t}B&=-k^2d_2B+\nu_1A+\nu_2B+\gamma_1A^2+\gamma_2A^3\;.
\end{align}
This is motivated by the form of various equations from Turing's own planned follow-up work on phyllotaxis \cite{dawesTuring,turingcol}, similar to the Swift-Hohenberg model \cite{sh77}, and various other models of biological potentials, chemical reactions, and crystal growth, such as \cite{j21fahad,wehner20,turzi09,thiele21,nosov16,kreusser20,winfree91} for example. We specifically choose nonlinear terms such that the righthand-side of \cref{pde00} is simple enough to analyze by hand, but is not a gradient field (so that something beyond standard catastrophe theory is required), and any nonlinear terms (including mixed terms like `$AB$') may be added in general. 

It will be convenient to make a linear transformation $(u,v)=(A+\sfrac{\gamma_1}{3\gamma_2},B+\sfrac{\alpha_1}{3\alpha_2})$ to remove the quadratic term. Let $\mu_i=\nu_i=0$ to simplify further (since we will still be left with a linear term from the diffusion), and let $\alpha_2=\gamma_2=-1$, then tidy up constants by defining $\alpha=\sfrac13\alpha_1^2$, $\gamma=\sfrac13\gamma_1^2$, $\beta=\sfrac13\gamma_1d_1-\sfrac2{27}\alpha_1^3$, $\delta=\sfrac13\alpha_1d_2-\sfrac2{27}\gamma_1^3$, $k_i=k^2d_i$, leaving us with
\begin{align}\label{pde2}
\sfrac{\partial\;}{\partial t}u&=-k_1u-\beta-\alpha v-v^3\;,\nonumber\\
\sfrac{\partial\;}{\partial t}v&=-k_2v-\delta-\gamma u-u^3\;,
\end{align}
for some diffusion constants $k_1,k_2,$ and reaction constants $\alpha,\beta,\gamma,\delta$. 

Our purpose now is to use this to illustrate the finding of underlying catastrophes. We will not explore the reaction diffusion problem thoroughly in any specific context, that is left to more exhaustive future work. For ease of discussion let us say that $u(t,x,y)<0<v(t,x,y)$ indicates $v$ (i.e. $B$) is dominant over $u$ (i.e. $A$), and vice versa, at some given $(t,x,y)$. 

Let us get some basic intuition for this system's dynamics, starting with its homogeneous steady states. This is purely for illustration so will not be complete. 

To look at the homogeneous system, set the diffusion constants to $k_1=k_2=0$. The righthand side is then divergence-free, equal to $-\nabla\times\phi$ with $\phi=\delta u+\beta v+\hf(\gamma u^2+\alpha v^2)+\sfrac14(u^4+v^4)$. At $\alpha=\beta=\gamma=\delta=0$ the righthand side is $(-v^3,-u^3)$, with a highly degenerate saddlepoint at the origin. 

Perturbing any of the coefficients can cause this to degenerate into up to 9 steady states. Diffusion given by $k_1,k_2>0$, for instance, stabilizes the origin, creating an attractor where the species $A$ and $B$ are `balanced' with concentrations $u=v=0$. Examples with and without diffusion are shown in \cref{fig:9pole}: without diffusion in (a) the origin is a saddlepoint, but with diffusion in (b) the origin becomes attracting. 

\begin{figure}[h!]\centering
\includegraphics[width=0.35\textwidth]{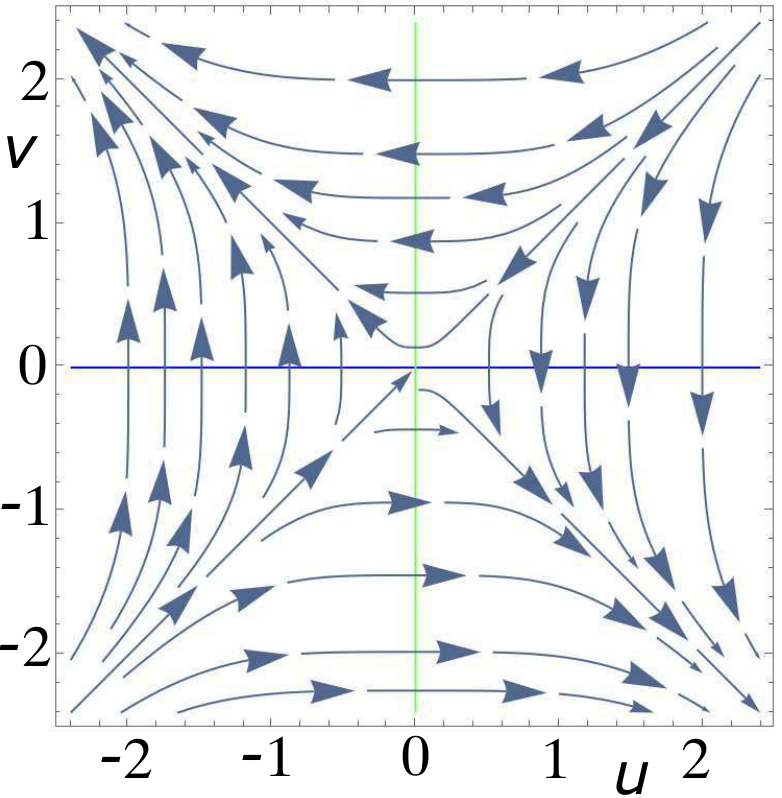}\hspace{1cm}
\includegraphics[width=0.35\textwidth]{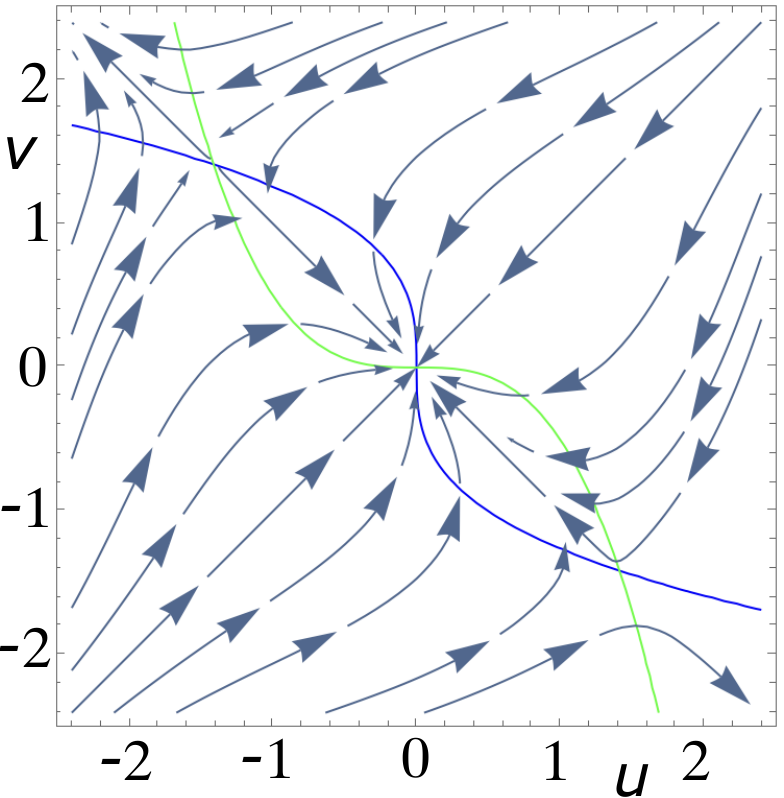}
\vspace{-0.3cm}\caption{\small\sf Flow in the $(u,v)$ plane with $\alpha=\beta=\gamma=\delta=0$. Left: with zero diffusion $k_1=k_2=0$ there is one saddle steady state at the origin. Right: with diffusion $k_1=k_2=2$ there is an attracting node at the origin, now surrounded by two saddles. }\label{fig:9pole}\end{figure}

For example, if we take $\beta=\delta=0$ with $\alpha>0$ and $\gamma>0$, the concentration steady states are given by $(u_{ij},v_{ij})=(i\sqrt{-\gamma},j\sqrt{-\alpha})$ with $i,j=-1,0,+1$. The states $(\pm\sqrt{-\gamma},0)$ and $(0,\pm\sqrt{-\alpha})$ are centres, therefore surrounded by states of oscillatory concentrations, and these are stabilized to the steady state by small diffusion $k_1,k_2>0$. If $\alpha$ and $\gamma$ have opposing signs, then only $(0,0)$ is a centre, again stabilized by small diffusion. These are shown with small diffusion in \cref{fig:poles}. 

\begin{figure}[h!]\centering
\includegraphics[width=0.35\textwidth]{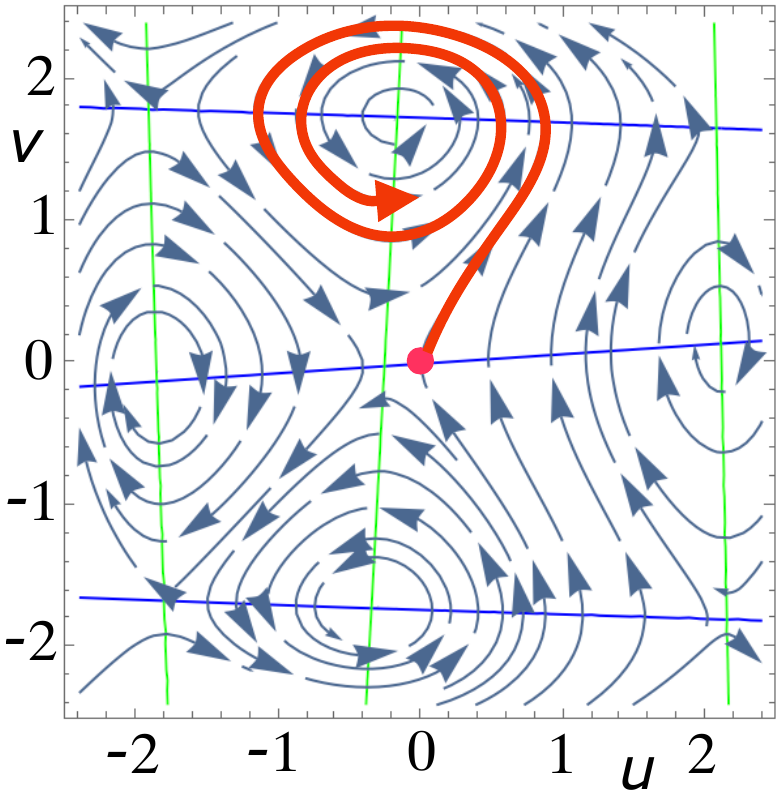}\hspace{1cm}
\includegraphics[width=0.35\textwidth]{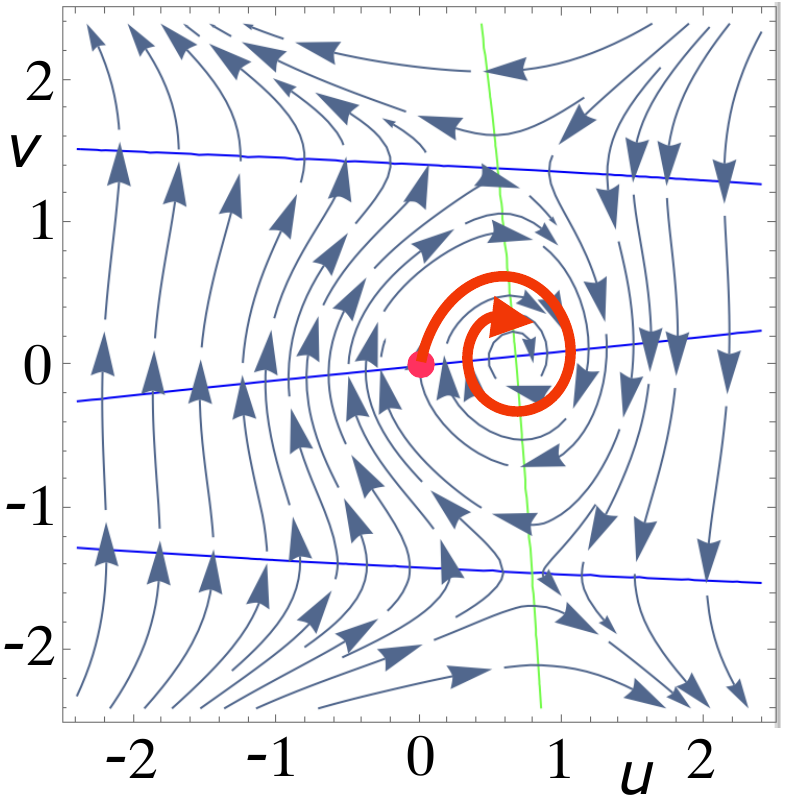}
\vspace{-0.3cm}\caption{\small\sf Flow in the $(u,v)$ plane, $k_1=k_2=0.2$, $\beta=0$, $\delta=-1$, $\alpha=-3$, and left $\gamma=-3$, right $\gamma=1$. Left: an initially balanced state (initial condition at the origin) would evolve to the state with $v$ dominant. Right: an initially balanced state would remain close to balance, gaining a small $u$ dominance. Solution from the balanced state shown in red.}\label{fig:poles}\end{figure}

Clearly there are various other parameter regimes that must be considered, with different numbers and stabilities of the steady states, so for a more complete picture we turn to their catastrophes. 

First calculate the determinants $\op B_i$ from \cref{def:sings} to the righthand side of \cref{pde2}, i.e. taking
%which we take as the vector field $\v F$. 
\begin{align*}%\label{pdeF}
\v F&=-(\;k_1u+\beta+\alpha v+v^3\;,\;k_2v+\delta+\gamma u+u^3\;)\;.
\end{align*}
The calculations are straightforward, and the first four evaluate as
\begin{align}%\label{}
\op B_1=&k_1k_2-ac\;,\\
\op B_{2}=&6(k_2ua-vc^2)\;,\\
\op B_{3}=&6(18k_2uvc-k_2^2a-c^3)\;,\\
\op B_{4}=&72k_2(3uc^2-2k_2vc-9k_2u^2v)\;,\\
{\rm where} &\quad\;\; a=\alpha+3v^2\;,\quad c=\gamma+3u^2\;.\nonumber
\end{align}

Where these vanish we will find sets of fold, cusp, swallowtail, and butterfly catastrophes, and one may solve to find their roots easily by computer. Or, more usefully, we can analyze them using the kind of parameterizations very beautifully expounded for the elementary catastrophes in \cite{ps96}. We do this by applying successive $\op B_r$ conditions to express the catastrophe sets in the space of the coefficients $(\alpha,\beta,\gamma,\delta)$, taking the variables $(u,v)$ as parameters. 

First parameterize the family of steady states, where \cref{pde2} vanishes, as 
\begin{align}\label{ss}
\beta(u,v)&=-k_1 u-\alpha v-v^3\;,\nonumber\\
\delta(u,v)&=-k_2 v-\gamma u-u^3\;.
\end{align}
One way to think of this is as a 4-dimensional manifold obtained by mapping from $(u,v;\alpha,\beta,\gamma,\delta)$ into $(\alpha,\beta,\gamma,\delta)$. 

Pairs of steady states will collide to form folds where $\op B_1$ vanishes, which we can parameterize as 
\begin{align}\label{b1}
0=\op B_1&=k_1k_2-(\alpha+3v^2)(\gamma+3u^2)\nonumber\\
\Rightarrow\qquad&
\alpha(u,v)=\sfrac{k_1k_2}{\gamma+3u^2}-3v^2\;.
\end{align}
These are a little easier to visualize, in the sense that they constitute a 3-dimensional manifold in the space of $(\alpha,\beta,\gamma,\delta)$, traced out by the functions \cref{ss} and \cref{b1} as $(u,v)$ and $\gamma$ vary. 

Different branches of folds will collide at cusps, given by
\begin{align}\label{b2}
0=\op B_{2}&\propto \sfrac{k_1k_2^2u}{\gamma+3u^2}-v(\gamma+3u^2)^2\nonumber\\
\Rightarrow\qquad&
\gamma(u,v)=(k_1k_2^2u/v)^{1/3}-3u^2\;, 
\end{align}
a 2-dimensional surface in the space of $(\alpha,\beta,\gamma,\delta)$ defined by the functions \cref{ss} to \cref{b2}, parameterized by $(u,v)$. 

These collide to form swallowtails. To find these it is useful to work in new variables $p=v/u$ and $q=uv$, and take $p$ as a parameter. Dropping some multiplicative constants from our calculation of $\op B_3$, we obtain
\begin{align}\label{b3}
0=\op B_{3}
&\propto 18k_2uv(k_1k_2^2u/v)^{1/3}-\sfrac{k_1k_2^3}{(k_1k_2^2u/v)^{1/3}}-k_1k_2^2u/v\nonumber\\
%&=k_2(k_1k_2^2u/v)^{1/3}( 18uv-k_1^{1/3}k_2^{2/3}(u/v)^{-2/3}-k_1^{2/3}k_2^{1/3}(u/v)^{2/3})\nonumber\\
%&\nonumber\\
\Rightarrow\qquad&
q(p)=\sfrac1{18}(k_1k_2)^{1/3}(k_2^{1/3}p^{2/3}+k_1^{1/3}p^{-2/3})\;,
\end{align}
where we let 
\begin{align}
u=(q/p)^{1/2}\;,\quad v=(pq)^{1/2}\;,\quad p=v/u\;,\quad q=uv\;.
\end{align}
The set $q(p)$ together with \cref{ss} to \cref{b2} gives us a curve of swallowtails in the space of $(\alpha,\beta,\gamma,\delta$). 

Finally these collide to form butterflies, given by
\begin{align}\label{b4}
0=\op B_{4}&\propto 3u(\gamma+3u^2)^2-2k_2v(\gamma+3u^2)-9k_2u^2v\;,\nonumber\\
\Rightarrow\qquad
0
%&=-6k_2^{-1}((k_1k_2^2u/v)^{1/3})^2+4(v/u)((k_1k_2^2u/v)^{1/3})+18uv\nonumber\\
%&=-6k_1^{2/3}k_2^{1/3}p^{-2/3}+4k_1^{1/3}k_2^{2/3}p^{2/3}+18q\nonumber\\
%%&=-6k_2^{1/3}p^{-2/3}+4k_2^{2/3}p^{2/3}+k_2^{2/3}p^{2/3}+k_2^{1/3}p^{-2/3}\nonumber\\
&=-5k_1^{2/3}k_2^{1/3}p^{-2/3}+5k_1^{1/3}k_2^{2/3}p^{2/3}\nonumber\\
\Rightarrow\qquad
p&=(k_1/k_2)^{1/4}\;,\quad q=\sfrac19(k_1k_2)^{1/2}\;.
\end{align}
Hence butterflies are found at isolated points in the space of $(\alpha,\beta,\gamma,\delta)$ where
%\begin{align}%\label{}
%&u=\sfrac13k_1^{1/8}k_2^{3/8}\;,\quad v=\sfrac13k_1^{3/8}k_2^{1/8}\;,\quad\\&
%\beta=-\sfrac{16}{27}k_1^{9/8}k_2^{3/8}\;,\quad
%\delta=-\sfrac{16}{27}k_1^{3/8}k_2^{9/8}\;,\quad
%\alpha=\sfrac23k_1^{3/4}k_2^{1/4}\;,\quad
%\gamma=\sfrac23k_1^{1/4}k_2^{3/4}\;,
%\end{align}
\begin{align}%\label{}
&
u=\pm\sfrac13(k_1k_2^3)^{1/8}\;,\quad
\beta=\mp\sfrac{16}{27}(k_1^3k_2)^{3/8}\;,\quad 
\alpha=\sfrac23(k_1^3k_2)^{1/4}\;,\nonumber\\&
v=\pm\sfrac13(k_1^3k_2)^{1/8}\;,\quad
\delta=\mp\sfrac{16}{27}(k_1k_2^3)^{3/8}\;,\quad
\gamma=\sfrac23(k_1k_2^3)^{1/4}\;.
\end{align}

In \cref{fig:bdplane} we plot the fold surface defined by \cref{ss}-\cref{b1} in the parameter space $(\beta,\delta)$ for some example values of $(\alpha,\gamma)$. 

\begin{figure}[h!]\centering
\includegraphics[width=0.95\textwidth]{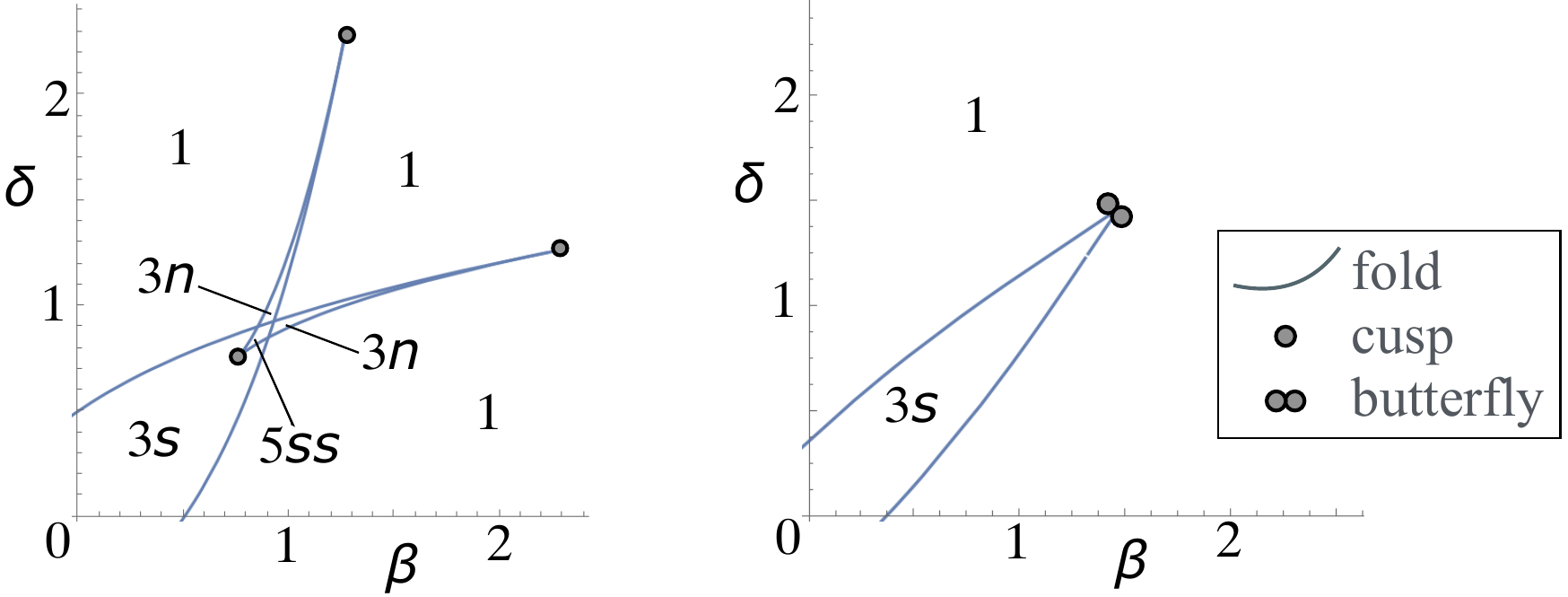}
\vspace{-0.3cm}\caption{\small\sf Stability in the $(\beta,\delta)$ plane, $k_1=k_2=1$. Left: $\alpha=\gamma=1/5$. Right: $\alpha=\gamma=2/3$. Numbers indicate how many steady states exist in each region, and letters indicate typical behaviour from an initially balanced concentration $u=v=0$: `s' indicates the system evolves to a unique stable steady state with $u$ or $v$ in slight dominance either side of $\beta=\delta=1$, `$ss$' indicates the system will evolve to one of two stable steady states where either $u$ or $v$ is dominant, `$n$' indicates there is a stable steady state that is not reachable from the origin so the concentrations diverge, and no letter indicates no stable steady states so the concentrations diverge. }\label{fig:bdplane}\end{figure}

In \cref{fig:kplane} we plot in the folds in the diffusion plane $(k_1,k_2)$ for some example values of $(\alpha,\beta,\gamma,\delta)$. The folds form the boundaries between regions that have different numbers of steady states as indicated, leading to different stable configurations for the system as described in the captions.

\begin{figure}[h!]\centering
\includegraphics[width=0.95\textwidth]{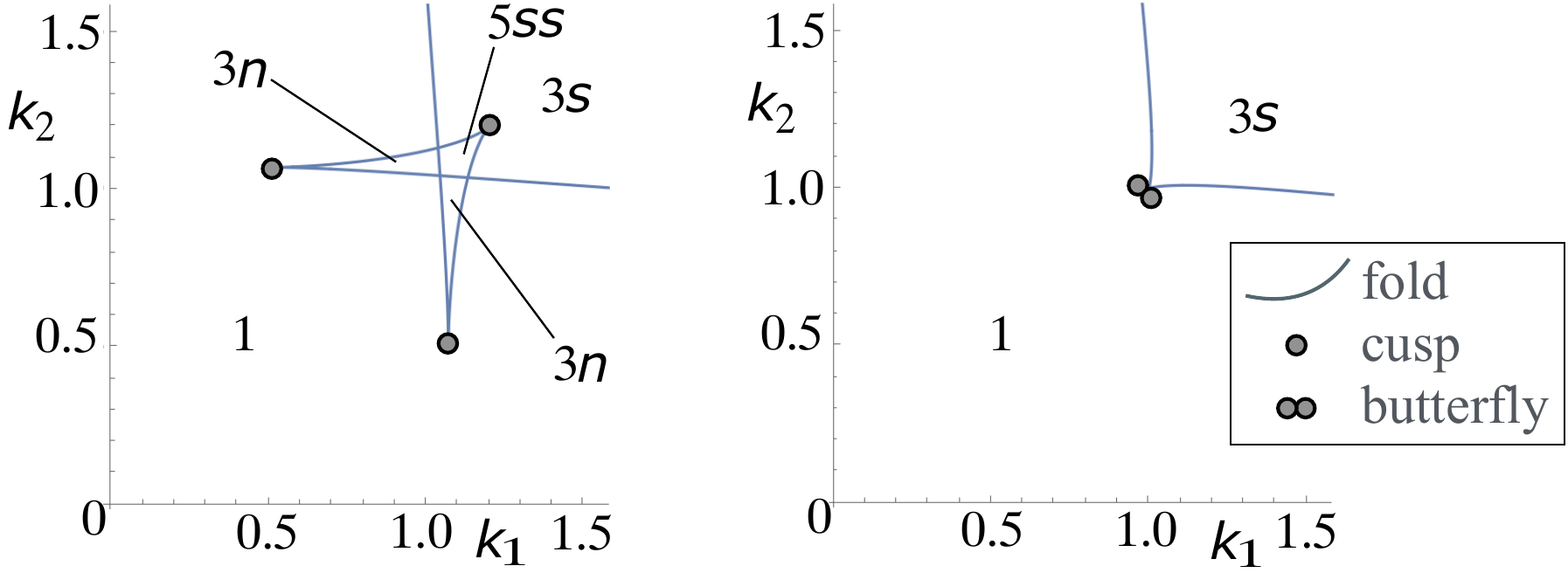}
\vspace{-0.3cm}\caption{\small\sf Stability in the $(k_1,k_2)$ diffusion parameter plane, with labels similar to \cref{fig:bdplane}. Left: $\alpha=\gamma=1/5$, $\beta=\delta=1$, the steady states lie at low concentrations. Right: $\alpha=\gamma=2/3$, $\beta=\delta=-16/27$, the steady states lie at high concentrations. }\label{fig:kplane}\end{figure}

These are just cursory examples to illustrate the concept. By giving appropriate initial data and boundary conditions for the differential equations, we could go on to simulate the mixing of the two-species in the $(x,y)$ plane at different times, or to study wave propagation through these concentrations. 
In \cite{j21fahad}, for example, underlying catastrophes of homogeneous steady states are used to determine the existence of heteroclinic connections that represent `wave-pinning' patterns in cell polarity formation in plants. A number of other works have attempted to relate elementary catastrophes to the behaviours of partial differential equations. 
In \cite{wehner20}, the signature shapes of a swallowtail are seen in bifurcation diagrams of an oxidation reaction, and in \cite{wehner23}, butterfly sets are then shown in reaction diffusion by limiting to a scalar one species problem, and a similar treatment places cusps in a crowd density problem in \cite{cuspjam10}. 
A nice discussion of the problem of detecting high codimension bifurcation points in PDEs can be found in \cite{kreusser20}, before scalar catastrophe theory is used to develop numerical methods for the purpose. 
An attempt is made to relate catastrophes to the phases of liquid crystals in \cite{turzi09} by focussing on a scalar term like the free energy. 

Typically in such studies it has been necessary to appeal to scalar functions to investigate the occurrence of catastrophes. The example above takes elements from these studies, and tries to show how catastrophes could be applied to them without reducing to scalar quantities. Hopefully this merely hints at the possible applications, and deeper study will require more detailed and specialised analysis with particular applications in mind.

Finally, but importantly, when deriving the catastrophes above we should have also calculated the non-degeneracy conditions $\op G$ from \cref{Gri} at each stage. Though straightforward, these are more numerous than the $\op B$ determinants, and best done with a computer. The relevant conditions are $\op G_1$ for the fold, $\op G_{2,1}$ and $\op G_{2,2}$ for the cusp, four determinants $\op G_{3,ij}$ for the swallowtail and eight determinants $\op G_{4,ijk}$ for the butterfly, where $i,j,k$ take values 1 and 2 (because the dimension of the system is 2). 

We will not give all of these calculations, but we find they are indeed non-vanishing for nonzero $k_1,k_2$. 
At the butterfly for example we find
\begin{align}%\label{}
\op G_{4,ijk}=103680(-1)^{i+k}(k_1k_2)^4(k_1/k_2)^{\sfrac34(k+2j+3i-9)}\
\end{align}
which is non-vanishing if and only if $k_1,k_2\neq0$, i.e. with diffusion.

The fact that these $\op G$ conditions are non-zero means the system is `full' according to \cref{def:bg}, meaning firstly that we can solve the conditions $\op B_i=0$, and moreover that these solutions are unique, i.e. we could not obtain different results by instead using the determinants from \cref{Bri}. Were any of these $\op G$s vanishing, the associated catastrophe would not be `full', and such cases require further case-specific study. 

Not being full may imply that the point has some higher degeneracy, for example we indeed find that $\op G_1$ vanishes at the set where the folds collide and become higher order catastrophes, i.e. the cusps, and the $\op G_{2,1}=\op G_{2,2}=0$ at the set where the cusps are degenerate and become swallowtails, etc.  

Not being full need always not imply an actual degeneracy of the system, however, as they may indicate the system has essentially become lower dimensional, so we just need to adjust the dimension of our calculations. If $k_1$ and $k_2$ vanish here then all of the $\op G$ conditions will vanish, and we can clearly see that the catastrophe calculations above all break down. What has actually happened in this case is that the righthand side of \cref{pde2} is now divergence-free, and expressible as the derivative of a scalar potential, to which we could apply elementary catastrophe theory instead. 

\section{Underlying catastrophes: the \BG~conditions}\label{sec:bg}

The \BG~determinants from \cref{def:sings} in \cref{sec:intro} are straightforward to evaluate even for large $n$ and $r$, so it is worth writing general expressions for any codimension. %, though they are a little more unwieldy. 

Strictly speaking our interest is only local, so we should refer to {\it germs} rather than {functions} $\v F$ or vector fields, but we are not really adding at all to the rigorous theory of bifurcations or singularities here, only developing practical methods to locate them, so we will introduce such technicalities only as needed. 

Recall that $\v F:\mathbb R^n\times\mathbb R^r\to\mathbb R^n$ is taken to be a smooth vector field with components $\v F=(f_1,f_2,...,f_n)$, and dependence written as $\v F=\v F(\v x,\alphab)$ for a variable $\v x\in\mathbb R^n$ and parameter $\alphab\in\mathbb R^r$.

We can give more general formulae for the \BG~determinants \cref{Br3}-\cref{Gr3} as follows. 
\begin{definition}\label{def:bg}
		For $i=1,2,...$ we first define
		\begin{align}\label{Br}
		\op B_{i}&=\abs{\frac{\partial(\op B_{i-1},f_2,...,f_{n})}{\partial(x_1,...,x_n)}}\;,
		\end{align}
		and for $i=0$ define $\op B_0\equiv f_1$. 
To generalize this to allow placing $\op B_{i_1}$ in the $k^{th}$ component of $\v F=(f_1,...,f_n)$, denote the vector $\v F$ with its $k^{th}$ row replaced by some scalar $V$ by
		\begin{align}\label{slash}
		(f_1,...,f_n)\backslash^kV=(f_1,...,f_{k-1},V,f_{k+1},...,f_n)\;,
		\end{align}
		then let
		\begin{align}\label{Bri}
		\op B_{i,K(i-1)}&=\abs{\frac{\partial(f_1,...,f_n)\backslash^{k_{i-1}}\op B_{i-1,K(i-2)}}{\partial(x_1,...,x_n)}}\;,
		\end{align}
		where $K(i)$ denotes an $i$-length string $K(i)=k_1...k_{i}$ of symbols $k_j\in\cc{1,...,n}$, for $j=1,...,i$. 
		We define the case $i=1$ as $\op B_{1,K(0)}=\op B_1$, and note that $\op B_{i,1...1}\equiv\op B_i$. 
		Finally let 
%		Finally, in terms of these we define extended determinants in the space $\mathbb R^n\times\mathbb R^r$, given by
		\begin{align}\label{Gri}
		\op G_{r,K(r-1)}&=\abs{\frac{\partial(f_1,...,f_{n},\op B_{1},\op B_{2,K(1)},...,\op B_{r,K(r-1)})}{\partial(x_1,...,x_n,\alpha_1,...,\alpha_{r})}}\;,
		\end{align}
		defining the case $r=1$ as $\op G_{1,K(0)}=\op G_1=\abs{\frac{\partial(f_1,...,f_{n},\op B_{1})}{\partial(x_1,...,x_n,\alpha_1)}}$. 
\end{definition}

		A crucial assumption was stated before \cref{def:sings}, namely that any set of $n-1$ of the gradient vectors $\sfrac{\partial f_i}{\partial x}$ are linearly independent. %, without which there may not exist a solvable set of conditions to locate a singularity at all. %This permits us to use only the functions $\op B_j$ from \cref{Br} in the conditions \cref{singr}.
		%It is possible that some grouping of $n-1$ the gradient vectors $\sfrac{\partial f_j}{\partial x}$, ..., $\sfrac{\partial f_n}{\partial x}$, are linearly dependent, in which case it is not generally possible to solve to locate a singularity in this way. 
		A useful way to express this property is with the following. 
\begin{definition}\label{def:subrank}
Let the {\bf subrank} of a function $\v F=\bb{f_1,...,f_n}$ be the least number of dimensions spanned by any $n-1$ of the gradient vectors $\sfrac{\partial f_1}{\partial\v x}$, ..., $\sfrac{\partial f_n}{\partial\v x}$, or using the notation \cref{slash}, 
\begin{align}
\subrank\v F
&=\min\bb{\operatorname{rank}\sq{(\sfrac{\partial f_1}{\partial\v x},...,\sfrac{\partial f_n}{\partial\v x})\backslash^j\v 0\;,\;\;j=1,...,n}}\;.
\end{align} 
\end{definition}
\noindent 
\Cref{def:sings} then applies to a vector field on $\mathbb R^n$ for which 
\begin{align}\label{subrankfull}
\subrank\v F=n-1\;.
\end{align}

		The property of being {\it full} in \cref{def:sings} guarantees that the conditions \cref{Br} are solvable to find $(\v x_*,\alphab_*)\in\mathbb R^n\times\mathbb R^r$, %and that $\v F$ is $\op K$-versal at $(\v x_*,\alphab_*)$, 
		as we discuss in \cref{sec:Gfull}. % and \cref{sec:K}, respectively. 
		
If \cref{subrankfull} is not satisfied then the determinants in \cref{Br} are not sufficient to locate a catastrophe, and one must instead solve all $\sfrac{1-n^r}{1-n}$ conditions $\op B_{i,K(i-1)}=0$, and there is no reason to expect these to be solveable in general unless a system has this many parameters.  
% as we will see when we recap the classification of singularities in \cref{sec:TB}. 

Even if \cref{subrankfull} is not satisfied, it is sometimes possible to trivially reduce the dimension $n$ of the system until \cref{subrankfull} {\it is} satisfied. We can illustrate this using the primary form of a catastrophe from \cref{primary}. 
%%In a following paper \cite{jc23tbprocedure} we intend to prove that the {\it underlying catastrophes} of \cref{def:sings} are in fact transformable, in the sense that they are {\it contact equivalent} as mappings, to the functions \cref{primary}, which constitute {\it pre-normal} forms of the singularities. 
The vector field in \cref{primary} has subrank $n-1$ provided $k_i\neq0$ for all $i=2,...,n$, such that at a singularity we have $\sfrac{\partial f_1}{\partial\v x}=\bb{0,k_2,...,k_n}$. %, ensuring that the rows $\sfrac{\partial f_i}{\partial\v x}$. 
If the subrank is less than $n-1$, however, then we can remove redundant dimensions and apply \cref{def:sings} to the reduced system. For example if $k_i=0$ for all $i>2$ in \cref{primary}, then $\subrank\v F(0)=1$ so we cannot use \cref{def:sings}, but if we reduce it to a planar system $\v F=\bb{f(x_1,\alphab)+k_2x_2,\lambda_2x_2}$ then \cref{def:sings} applies. 

%We shall show in \cref{sec:B} that the conditions \cref{singr} are equivalent to $\v F$ having a {\it Morin} singularity, with Boardman symbol $\tb=1,\stackrel{r\;{\rm times}}{\dots},1$. Moreover we show in \cref{sec:G} that such a codimension $r$ catastrophe defines a point of $\v F$ that degenerates into up to $r+1$ zeros of $\v F$ under perturbation. 

%\newpage
%%%%%%%%%%%%%%%%%%%%%%%%%%%%%%%%%%%%%%%%%%%%%%%%%%%%%%
%%%%%%%%%%%%%%%%%%%%%%%%%%%%%%%%%%%%%%%%%%%%%%%%%%%%%%
\section{\BG~determinants and Boardman symbols}\label{sec:B}

The idea set out in \cite{j22cat} essentially makes three propositions, namely that for a vector field $\v F:\mathbb R^n\times\mathbb R^r\to\mathbb R^n$  in variables $(x_1,...,x_n)$ and parameters $(\alpha_1,...,\alpha_r)$:
\begin{enumerate}
\item[(i)] Any local bifurcation point of $\v F$, namely where $\v F=0$ and its Jacobian has less than full rank, has an {\it underlying catastrophe}. % (see \cref{def:sings} in \cref{sec:bg}). %; loosely these are simply points where a zero of $\v F$ coincides with a singularity of $\v F$. 
%\item[(ii)] %corank $1$ (along with further restrictions making this a {\it Morin} singularity in the Thom-Boardman classification for mappings). 
\item[(ii)] An {\it underlying catastrophe} is a degenerate zero of $\v F$ that bifurcates under perturbation into a perturbation-dependent number of zeros. In particular, a corank 1 {\it underlying catastrophe} %, defined by conditions $\v F=\op B_1=\op B_2=...=\op B_r=0\neq\op G_r$, 
of codimension $r$ would break up under perturbation into up to $r+1$ zeroes of $\v F$.
\item[(iii)] For certain {\it primary forms} $\bb{f(x_1,\alpha_1,...,\alpha_r),x_2,...,x_n}$ in terms of a scalar polynomial $f:\mathbb R\times\mathbb R^r\to\mathbb R$ (eq.\cref{primary} in \cref{sec:intro}), the conditions defining an underlying catastrophe reduce to the familiar defining conditions of Thom's {\it elementary catastrophes}. 
\end{enumerate}
These properties are the motivation for the term `underlying catastrophes', and we will show here the sense and conditions under which they are indeed true. 
The term `underlying' comes from (i)-(ii), namely that we are identifying bifurcations only by looking for places where zeroes of $\v F$ encounter singularities, without any regard for the vectorial nature of $\v F$, so in a sense we are looking only at the mapping `underlying $\v F$' (ignoring that it maps to a tangent space and this carries consequences for local stability). For cases of corank 1, giving us what Arnold classifies as series $A_r$ with “one zero eigenvalue and $(r-1)$-fold degeneracy in the nonlinear terms” \cite{a93v}, we will show here that the singularity classification reduces precisely to the conditions $\op B_i=0$. 
The use of the term `catastrophe' is then due to (iii), namely that these conditions then reduce further to the identifying conditions of the elementary catastrophes for the primary forms \cref{primary}. 
%Accompanying work in \cite{jc23tbprocedure} will further show that if $\v F$ has a corank 1 {underlying catastrophe}, then it is in fact transformable into such a primary form, though considered only as a map, and not as a vector field. 

%Take a smooth map $\v F:\mathbb R^n\to\mathbb R^n:x\mapsto\bb{f_1(x),...,f_n(x)}$. We are interested in points where $\v F=0$, and moreover where the Jacobian matrix $\nabla\v F$ has corank 1 (i.e. rank $n-1$). 
%
%Let's have a bash at the Boardman symbols without introducing ideals, then we'll just bring them in later if/when we need them for equivalence. 

Ren\'e Thom introduced a general classification of singularities in \cite{thom55} for mappings $\v F:\mathbb R^m\to\mathbb R^n$. They were defined through the transversality of certain submanifolds defined by kernels of $\v F$ and its derivatives. John Boardman provided a more explicit construction in \cite{boardman67}, developed on by Mather and Morin \cite{mather73,morin75}, which can be expressed as a sequence of extended Jacobian matrices containing minors of successive derivatives of $\v F$, culminating in a convenient set of {\it Boardman symbols} for Thom's singularities. % (see e.g. \cite{mond20}). 

The Thom-Boardman singularity classification is less well known than perhaps it should be among mathematicians, and after half a century is still far less well utilised in applications than it could be. So let us begin with a slightly informal, methodological introduction to the Thom-Boardman procedure, before showing how the \BG~conditions not only fit into the Thom-Boardman theory, but vastly simplify it for a significant class of singularities. 
We will restrict our attention to $\v F:\mathbb R^n\to\mathbb R^n$ but the extension for $\v F:\mathbb R^m\to\mathbb R^n$ with $m\neq n$ is straightforward. 
Where possible we leave the more rigourous statements behind the theory to references. %, but we do re-introduce some vital elements of the formalism later on, where they are necessary to establish generality. 

\bigskip
%%%%%%%%%%%%%%%%%%%%%%%%%%%%%%%%%%%%%%%%%%%%%%%%%%%%%%
\subsection{Thom-Boardman singularities}\label{sec:TB}	%, the ideal-free version

The Thom-Boardman classification makes heavy use of {\it ranks} and {\it minors} of certain Jacobian matrices. 
The $k\times k$ minors of a $p\times q$ matrix $M$ are the determinants of all $k\times k$ matrices taken from $M$ by omitting any $p-k$ rows and $q-k$ rows. 
A $p\times q$ matrix $M$ has rank $s$ if there exists an $s\times s$ minor of $M$ that is non-zero, but every larger minor {\it is} zero. 

Let $\nabla$ denote the derivative operator with respect to $\v x=(x_1,...,x_n)$. 
We will be dealing with vectors and matrices of differing sizes below, so for some $\v F=(f_1,...,f_n)$ and $\v G=(g_1,...,g_m)$, for any $n,m,$ let us be able to combine them as $(\v F,\v G)=(f_1,...,f_n,g_1,...,g_m)$, with Jacobians $\nabla(\v F,\v G)=(\nabla\v F,\nabla\v  G)=(\nabla f_1,...,\nabla f_n,\nabla g_1,...,\nabla g_m)$. We refer to each component $f_i$ or $g_i$ as occupying a {\it row} of $(\v F,\v G)$, and similarly each $\nabla f_i$ or $\nabla g_i$ occupies a {\it row} of $\nabla(\v F,\v  G)$. 

To define the {\it Boardman symbol} of a singularity we can define a sequence of extended Jacobians of $\v F$. 

Given $\v F=(f_1,...,f_n):\mathbb R^n\to\mathbb R^n$, first define the symbol $\Delta^1\v F=\bb{f_1,...,f_n,\det\nabla\v F}$, and also for $i_1>1$ define $\Delta^{i_1}\v F=\bb{\v F,m^{i_1}_1,m^{i_1}_2,...,}$, where $m^{i_1}_j$, $j=1,...,N_1$, are the $(n-{i_1}+1)\times(n-{i_1}+1)$ minors of $\nabla\v F$. %We give a formula for $N_1$ in \cref{sec:count}. 
 
Next define $\Delta^{i_2}\Delta^{i_1}\v F=\bb{\Delta^{i_1}\v F,m^{i_2}_{2,1},m^{i_2}_{2,2},...}$, where $m^{i_2}_{2,j}$ are the $(n-i_2+1)\times(n-i_2+1)$ minors of $\nabla(\Delta^{i_1}\v F)$, and proceed iteratively, next defining $\Delta^{i_3}\Delta^{i_2}\Delta^{i_1}\v F=\bb{\Delta^{i_2}\Delta^{i_1}\v F,m^{i_3}_{3,1},m^{i_3}_{3,2},...}$, where $m^{i_3}_{3,j}$ are the $(n-i_3+1)\times(n-i_3+1)$ minors of $\nabla(\Delta^{i_2}\Delta^{i_1}\v F)$, and so on, giving generally
\begin{align}
\Delta^{i_j}...\Delta^{i_1}\v F=\bb{\Delta^{i_{j-1}}...\Delta^{i_1}\v F,m^{i_j}_{j,1},m^{i_j}_{j,2},...}\;,
\end{align}
where $m^{i_j}_{j,k}$ for $k=1,...,N_j,$ are the $(n-i_j+1)\times(n-i_j+1)$ minors of $\nabla(\Delta^{i_{j-1}}...\Delta^{i_1}\v F)$. We give a formula for the number $N_j$ of minors in each stage of these calculations in \cref{sec:count}.

%\noindent{\\\\\bf From David:\\ \it The formal definition of  $T-B_r$ (let's call it) is that an ascending sequence of ideals (of certain functions) each has rank  $n-1$  until at last you hit one with rank  n .  However, the rank of an ideal is by definition the rank of the matrix whose rows are the gradients of those functions, and, as we all know, to evaluate the rank you have to evaluate a whole bunch of  $n\times n$  determinants (minors) - namely, your  B-determinants.\\}

\begin{definition}\label{def:TBsymbol}
The {\it Boardman symbol} of $\v F$ at $\v x=0$ is the sequence $\tb=\tb_1,...,\tb_r,$ such that each $\nabla\Delta^{\tb_{j-1}}...\Delta^{\tb_1}\v F(0)$ has corank $\tb_j$ for $j=1,...,r$, (including that $\nabla\v F(0)$ has corank $\tb_1$), and the symbol $\tau$ is taken to terminate at $r$ such that $\tb_{r+1}=0$. 
\end{definition}

We call $\tb_j$ the $j^{th}$ Boardman symbol. 

Note that the symbols $\tb_1,...,\tb_r$, form a non-increasing sequence $\tb_1\ge\tb_2\ge...\ge\tb_r$. 
This means, for instance, that if the Jacobian of $\v F$ has corank one, then the Boardman symbol consists of a sequence of 1s only, and in this case we show, in the next section, that the minors $m_{j,k}^{i_j}$ are reducible to precisely the functions $\op B_{i,K(i-1)}$ from \cref{def:bg}.

The description above is usually made not in terms of functions $\Lambda^{\tb_r}...\Lambda^{\tb_1}\v F$, or even their germs (loosely their `local representations'), but in terms of the ideals (loosely `linear combinations') generated by the components of those germs, see e.g. \cite{gibson79,mond20}. But the procedure is equivalent, and since the ranks of ideals are actually defined by referring back to the ranks of the functions $\nabla\Lambda^{i_r}...\Lambda^{i_1}\v F$, the above seems both more direct and less unnecessarily technical, at least for the purposes of describing the calculations involved. A more technically complete presentation will be given in forthcoming work \cite{jc23tbprocedure}.
%The technical rigour of ideals will become indispensible later, in \cref{sec:G}, to show why this classification describes not just individual functions but equivalence classes. 

\bigskip
%%%%%%%%%%%%%%%%%%%%%%%%%%%%%%%%%%%%%%%%%%%%%%%%%%
\subsection{The case $\tb=1,...,1$, and the $\op B$ determinants}\label{sec:morin}

Let us consider more closely what these minors look like if $\nabla\v F$ has corank one. In this case the Boardman symbols are just a sequence of symbols $\tb_1=...=\tb_r=1$, $\tb_{r+1}=0$, for some $r\ge1$. These are known as $\Sigma^1$ singularities or {\it Morin singularities} \cite{mond20,morin75}.

We will show below that, if we assume a codimension $r$ Morin singularity at some $\v x=\v x_0$, we can locate it by either:
\begin{enumerate}
\item[(a)] solving to find where the requisite minors $m^{\tb_j}_{j,k}$ associated with the Boardman symbols $\tb_1=...=\tb_r=1$ all vanish at some $\v x=\v x_0$. The number of these minors grows super-factorially with $r$. 
\item[(b)] solving the conditions $\op B_1=\op B_{2,k_1}=...=\op B_{r,k_1...k_{r-1}}=0$ for all $k_1,...,k_{r-1}\in1,...,n,$ at $\v x=\v x_0$. There are $\sum_{m=0}^{r-1}n^m=\sfrac{1-n^r}{1-n}$ of these conditions. 
\item[(c)] restricting to vector fields with $\subrank\v F=n-1$ (\cref{def:subrank}), and solving the conditions $\op B_1=...=\op B_r=0$ at $\v x=\v x_0$. There are only $r$ of these conditions, and the $\op G$ conditions (\cref{def:sings}) guarantee these have isolated roots. 
\end{enumerate}
In (a)-(b) the number of conditions can be expected to exceed the number of available parameters in a typical physical model, so they do not typically give a solveable system of equations. This in contrast to (c) which requires solving $r$ equations in $r$ parameters, so while the \BG~conditions cannot detect all possible singularities, they restrict attention to those that can be solved for in closed form.

To show this we have the following series of results, which we state, and then give their proofs afterwards. Recall the definitions of the determinants $\op B_i$ and $\op B_{i,K(i-1)}$ from \cref{Br}-\cref{Bri}. The statements ``all $K(i-1)$'' or ``all $\op B_{i,K(i-1)}$'' are shorthand for writing that the string $K(r-1)=k_1...k_{r-1}$ or function $\op B_{i,k_1...k_{i-1}}$ are taken for all $k_j=1,...,n$ and for $j=1,...,r-1$. 

Below we will want the dependence on the variable $\v x\in\mathbb R^n$ and parameter $\alphab\in\mathbb R^r$ to be explicit, so write
$$\v F=\v F(\v x;\alphab)\;,\qquad{\rm and}\qquad \op B_{i}=\op B_{i}(\v x;\alphab)\;.$$
The results below assume, for $\v F:\mathbb R^n\to\mathbb R^n$, that 
\begin{align}\label{0span}
\v F(0;0)=0\qquad{\rm and}\qquad \subrank\v F(0;0)=n-1\;.
\end{align}

First we will need to show that if one choice of $\op B_i$ vanishes, then all the permutations $\op B_{i,K(i-1)}$ vanish. This is given by the following. 
\begin{lemma}[Vanishing of $\op B_i$]\label{thm:Br}
	If $\op B_1(0;0)=0$ and $\op B_{2,j}(0;0)=0$ for some $j$, then $\nabla\op B_{2,k}(0;0)=0$ for all $k=1,...,n$. 
\end{lemma}
\begin{corollary}\label{thm:Brall}
	If $\op B_1(0;0)=\op B_{2}(0;0)=...=\op B_{r}(0;0)=0$, then $\op B_1(0;0)=\op B_{2,K(1)}(0;0)=....=\op B_{r,K(r-1)}(0;0)=0$ for all $K(1),...,K(r-1)$, hence we can reduce each of the $\op B_{i,K(i-1)}$ to one choice $\op B_{i}$ for $i=1,...,r$. 
\end{corollary}

We will further need to show that the vanishing of $\op B_i$ implies vanishing of all of the minors involved in the $i^{th}$ Boardman symbol. To show this we will make use of the following. 
\begin{lemma}[Vanishing minors]\label{thm:minors}
	If $\op B_1(0;0)=0$ and all $\op B_{r,K(r-1)}(0;0)=0$ for some $r>1$, then any of the $n\times n$ minors of any $N\times n$ matrix formed from $N>n$ rows of the form $\nabla\op B_{r-1,K(r-2)}(0;0)$ must vanish. 
\end{lemma}
\begin{corollary}\label{thm:minorsall}
	If all $\op B_1(0;0)=...=\op B_{r,K(r-1)}(0;0)=0$, then any of the $n\times n$ minors of any $N\times n$ matrix formed from $N>n$ rows $$\nabla F_j(0;0),\;\nabla\op B_1(0;0),\;...,\;\nabla\op B_{r-1,K(r-2)}(0;0)\;,$$ must vanish. 
\end{corollary}

Lastly we will need the converse of \cref{thm:Br}, namely to say that $\op B_i=0$ {\it if and only if} all $\op B_{i,K(i-1)}=0$. We want to show this in the extended space $\mathbb R^n\times\mathbb R^r$. 
\begin{lemma}[Uniqueness of $\op B_i$]\label{thm:Buni}
The condition $\op G_{r,K(r-1)}(0;0)\neq0$ for all $K(r-1)=k_1...k_{r-1}$, with the assumption \cref{0span}, implies that $\op B_i(0;0)=0$ {\it if and only if} all $\op B_{i,K(i-1)}(0;0)=0$ for $i=1,...,r$. 
\end{lemma}

These allow us to prove the main result. 
\begin{theorem}[Underlying catastrophes and Thom-Boardman symbols]\label{thm:TBB}
	An underlying catastrophe of codimension $r$, given by \cref{def:sings} with $\op B_1(0;0)=...=\op B_{r}(0;0)=0$, is equivalent to a Thom-Boardman singularity with Boardman symbol $\tb=1,...,1$ (of length $r$). 
\end{theorem}

\bigskip
%%%%%%%%%%%%%%%%%%%%%%%%%%%%%%%%%%%%%%%%%%%%%%%%%%
\subsection{Proof of the above}\label{sec:proofs}

Denote the linear subspace spanned by the rows of $\nabla\v F(0;0)$ as $\op T\subset\mathbb R^n$. 
%, and the linear subspace spanned by the rows of $\nabla\bb{\v F(0;0)\backslash^j V(0;0)}$ as $\op T_{V(0;0)}^j\subset\mathbb R^n$. 

First take the trivial case $r=0$, so there is no Boardman symbol (or the symbol is just 0), in particular the first Boardman symbol is trivially $\tb_1=0$. So the rank of $\nabla\v F(0;0)$ is $n$, i.e. $\op B_1(0;0)=\det\nabla\v F(0;0)\neq0$ and there is no singularity at the origin, and hence the subspace $\op T$ is $n$-dimensional. 

\medskip

Then let us start with two simple and generally known results that we will use below. 
\begin{lemma}[Singularity]\label{thm:sing}
	If $\op B_1(0;0)=0$ and $\subrank\v F(0;0)=n-1$, the linear subspace spanned by the rows of $\v F(0;0)$ is an $(n-1)$-dimensional hyperplane $\op T\subset\mathbb R^n$. 
\end{lemma}
\begin{proof}[Proof of \cref{thm:sing}]
If $\op B_1(0;0)=\det\nabla\v F(0;0)=0$ and $\subrank\v F(0;0)=n-1$, then $\nabla\v F(0;0)$ has corank 1 (and hence defines a singularity with first Boardman symbol $\tb_1=1$), 
therefore the subspace $\op T$ is $(n-1)$-dimensional. 
\end{proof}

\begin{lemma}[Equivalence of subspaces]\label{thm:bigspan}
Let $M=\nabla\v F$ and $\tilde M_j=\nabla(\v F\backslash^jg)=\bb{\nabla f_1,...,\nabla f_{j-1},\nabla g,\nabla f_{j+1},...,\nabla f_n}$, where $\v F=\bb{ f_1,..., f_n}$ for smooth functions $f_i:\mathbb R^n\to\mathbb R^n$ and $g:\mathbb R^n\to\mathbb R^n$. Let $M(0;0)$ and $\tilde M_j(0;0)$ have corank 1 and let $\subrank\v F(0;0)=n-1$. Then the rows of $M(0;0)$ and $\tilde M(0;0)$ span the same $(n-1)$-dimensional subspace $\op T\subset\mathbb R^n$. 
\end{lemma}
\begin{proof}[Proof of \cref{thm:bigspan}]
Since $M(0;0)$ and $\tilde M_j(0;0)$ both have corank 1, their rows span $(n-1)$-dimensional subspaces, say $\op T\subset\mathbb R^n$ and $\tilde{\op T}_j\subset\mathbb R^n$, respectively. By \cref{def:subrank}, $\subrank\v F(0;0)=n-1$ means any $n-1$ of the gradient vectors $\nabla f_1$, ..., $\nabla f_n$, omitting $\nabla f_j$ for some $j\in\cc{1,...,n}$, also span a subspace of exactly $(n-1)$ dimensions, so let us call this $\op T'_j$. Now the set $\op T'$ spanned by any $n-1$ of the gradient vectors $\nabla f_1$, ..., $\nabla f_n$, must be a subset both of the set $\op T$ spanned by $M(0;0)$, and the $j^{th}$ set $\tilde{\op T}_j$ spanned by $\tilde M_j(0;0)$, that is $\op T'\subseteq\op T$ and $\op T'\subseteq\tilde{\op T}_j$, but since each of these is $(n-1)$-dimensional this implies $\op T'=\tilde{\op T}_j=\op T$ for every $j$. 
\end{proof}

We use the dimensionality of $\op T$ to prove the main lemmas. First is to show that the vanishing of $\op B_i$ implies the vanishing of all permutations $\op B_{i,K(i-1)}$. We begin with $i=1$. 

\begin{proof}[Proof of \cref{thm:Br}]
Since $\op B_{2,k}=\v F\backslash^k\op B_1$, then if $\op B_{2,k}(0;0)=0$ for some $k$ this implies $\nabla\op B_{1}(0;0)$ spans an $(n-1)$-dimensional linear subspace with $\cc{\nabla f_2,...,\nabla f_n}$, and by \cref{thm:bigspan} this subspace must be $\op T$, i.e. the same linear subspace spanned by the rows of $\nabla\v F(0;0)$. Since \cref{thm:Br} also assumes $\op B_1(0;0)=0$, \cref{thm:sing} applies, hence $\nabla\op B_{2,k}(0;0)=0$ for all $k=1,...,n$. 
\end{proof}

Then extend this to $i>1$. 

\begin{proof}[Proof of \cref{thm:Brall}]
If $\op B_1(0;0)=\op B_{2}(0;0)=...=\op B_{r}(0;0)=0$, then applying \cref{thm:Br} at each order $i=2,...,r,$ implies that $\op B_1(0;0)=\op B_{2,K(1)}=...=\op B_{r,K(r-1)}(0;0)=0$ for every $K(1)\in[1,n]$, ..., $K(r-1)\in[1,n]^{r-1}$% (noting $K(1)$ is just the single index $k_1$)
. Hence at each order we can reduce $\op B_{i,K(i-1)}$ to one choice $\op B_{i}$ for $i=1,...,r$. 
\end{proof}

The next step is to relate this to the Boardman symbols by showing that if $\op B_i$ vanishes, then so do the numerous minors in the Thom-Boardman calculations. Again we show this first for $i=1$. 

\begin{proof}[Proof of \cref{thm:minors}]
	If $\op B_{r,K(r-1)}(0;0)=0$ for all $K(r-1)\in[1,...,n]^{r-1}$, then since $\op B_{r,k_1...k_{r-2}k_{r-1}}=\det(\nabla\v F\backslash^{k_{r-1}}\nabla\op B_{r-1,k_1...k_{r-2}})$, this implies by \cref{thm:bigspan} that we have $\nabla\op B_{r-1,K(r-2)}(0;0)\in\op T$ for all $K(r-2)=k_1...k_{r-2}\in[1,n]^{r-2}$. 
	\Cref{thm:minors} also assumes $\op B_1(0;0)=0$, so \cref{thm:sing} applies, so any $N\times n$ matrix formed from $N>n$ rows $\nabla\op B_{r-1,K(r-2)}(0;0)$ has rank at most $n-1$, and any of its $n\times n$ minors vanish. 
\end{proof}

Then extend this to $i>1$. 

\begin{proof}[Proof of \cref{thm:minorsall}]
	If all $\op B_1(0;0)=...=\op B_{r,K(r-1)}(0;0)=0$, then all $\nabla B_1(0;0)=...=\nabla\op B_{r-1,K(r-2)}(0;0)\in\op T$, and moreover \cref{thm:sing} applies, so any $N\times n$ matrix formed from $N>n$ rows $$\nabla F_j(0;0),\nabla\op B_1(0;0),...,\nabla\op B_{r-1,K(r-2)}(0;0),$$ has rank at most $n-1$, and any of its $n\times n$ minors vanish. 
\end{proof}

It is slightly longer to prove the converse to \cref{thm:Br}.  
\begin{proof}[Proof of \cref{thm:Buni}]
Denote the gradient derivative on $\mathbb R^n\times\mathbb R^r$ as $\square=(\nabla_{\v x},\nabla_\alpha)$. 
Say $\v F(0;0)=\op B_1(0;0)=0$ and $\det(\square\v F(0;0),\square\op B_1(0;0))\neq0$, so, by the inverse function theorem, the root $(0,0)$ is isolated. 

Now say $\v F(0;0)=\op B_1(0;0)=...=\op B_{r,1...1}(0;0)=0$, then \cref{thm:Br} implies $\op B_{i,K(i-1)}(0;0)=0$ for all $K(i-1)=k_1...k_{i-1}$. Assume $\op G_{r,K(r-1)}(0;0)\neq0$, i.e. $\det(\square\v F(0;0),\square\op B_1(0;0),\square\op B_{2,k_1}(0;0),...,\square\op B_{r,k_1...k_{r-1}}(0;0))\neq0$ for all $k_1,...,k_{r-1}$, so, by the inverse function theorem, the root $(0,0)$ is isolated in $(\v x,\alphab)$ space, and hence these roots must be the same or all $k_1,...,k_{r-1}$, hence $\op B_i(0;0)=0$ {\it if and only if} all $\op B_{i,K(i-1)}(0;0)=0$ for $i=1,...,r$, local to the origin. 

\end{proof}  

%This last proof can be stated quite succinctly in terms of ideals. Let $\op E_n$ denote the set of real-valued smooth germs in a neighbourhood of the origin in $\mathbb R^n\times\mathbb R^r$, and $\op M_n\subset\op E_n$ be the maximal ideal which is the subset of germs that vanish at the origin. Let $I$ be the ideal generated by the functions $\cc{F_1,...,F_n,\op B_1,...,\op B_r}$, which all vanish at the origin. Then $\op G_{r,1...1}(0;0)\neq0$ means $I$ has full rank $n+r$, hence $I=\op M_n$. Now, by \cref{thm:Brall}, all functions $\cc{F_1,...,F_n,\op B_1,\op B_{2,k_1}...,\op B_{r,k_1...k_{r-1}}}$ vanish at the origin for any $k_1...k_{r-1}$, so let these generate an ideal $I_{k_1...k_{r-1}}$, then $\op G_{r,k_1...k_{i-1}}(0;0)\neq0$ means $I_{k_1...k_{r-1}}$ has full rank $n+r$, hence $I_{k_1...k_{r-1}}=I=\op M_n$. %Though succinct, this is less explicit about the calculations involved when applying these ideas to characterize singularities. 

This brings us to the final proof of \cref{thm:TBB}, where we find precisely how the conditions $\op B_i=0$ appear (and are repeated many times over) in the procedure to obtain the Boardman symbols. To make this explicit we will proceed one codimension at a time through \cref{def:TBsymbol}.

%%%%%%%%%%%%%%%%%%%%%%%%%%%%%%%%%%%%%%%%%%%%%%%%%%%
%%%%%%%%%%%%%%%%%%%%%%%%%%%%%%%%%%%%%%%%%%%%%%%%%%%
\subsection{Proof of \cref{thm:TBB}}\label{sec:TBBproof}

%\begin{proof}[Proof of \cref{thm:TBB}]
First take the trivial case $r=0$, so there is no Boardman symbol (or the symbol is just 0), this means $\tb_1=0$, so the rank of $\nabla\v F(0;0)$ is $n$, i.e. $\op B_1(0;0)=\det\nabla\v F(0;0)\neq0$ and there is no singularity at the origin. %If $\v F(0;0)=0$ then $\Delta^1\v F=\bb{f_1,...,f_n,\det\nabla\v F}$ at the origin contains at least one nonzero component $\det\nabla\v F(0;0)\neq0$. 
%The subspace $\op T$ is $n$-dimensional. 

Now assume the first Boardman symbol is $\tb_1=1$. Then $\nabla\v F(0;0)$ has corank 1 and defines a singularity, and $\op B_1(0;0)=\det\nabla\v F(0;0)=0$. If the second Boardman symbol is $\tb_2=0$ so the complete symbol is just $\tb=1$, and defines a {\it fold}, then the $(n+1)\times n$ matrix $\nabla\Delta^1\v F(0;0)=\bb{\nabla\v F(0;0),\nabla\op B_1(0;0)}$ has rank $n$, so at least one of its $n\times n$ minors must be nonzero; we will look more closely at those minors in the next step. 
%Since $\op B_1(0;0)=0$ the subspace $\op T$ is $(n-1)$-dimensional. %Note however that since $\op B_{2,j}(0;0)\neq0$ for some $j$, the gradient $\nabla\op B_1(0;0)$ lies outside of $\op T$, so $\op T_{\op B_1}^j$ is $n$-dimensional for that $j$. 
%If $\op G_{1}\neq0$ $\op T_{\op B_1}$ is $n+1$-dimensional. 

Assume instead that $\tb_1=\tb_2=1$. Then $\nabla\v F(0;0)$ has corank 1, and moreover $\nabla\Delta^1\v F(0;0)$ has corank 1, so all of the $n\times n$ minors of $\nabla\Delta^1\v F(0;0)$ are zero. 
Those minors (with the exception of the minor $\nabla\v F(0;0)$ which we already know vanishes) are precisely the functions we define as $\op B_{2,k}$ for $k=1,...,n$. Let 
\begin{align}\label{minors2}
m^{1}_{2,k}=\op B_{2,k}(0;0)\quad{\rm for} \quad k=1,...,n,\quad{\rm and} \quad m^{1}_{2,n+1}=\op B_1(0;0)\;.
\end{align} 
%Hence $\op B_{2,j}(0;0)=0$ for all $j=1,...,n$. So 
Hence if $\tb_1=\tb_2=1$ then $\op B_1(0;0)=\op B_{2,k}(0;0)=0$ for all $k=1,...,n$. Conversely, if $\op B_1(0;0)=\op B_2(0;0)=0$, then by \cref{Br} we have $\op B_{2,k}(0;0)=0$ for all $k=1,...,n$, implying $\tb_1=\tb_2=1$. 

Now if $\tb_3=0$ we are done and the singularity is a {\it cusp}, then the $2(n+1)\times n$ matrix $\nabla\Delta^1\Delta^1\v F=\cc{\nabla\v F,\nabla\op B_1,\nabla\op B_{2,1},...,\nabla\op B_{2,n},\nabla\op B_1}$ must have rank $n$, implying that at least one of the $n\times n$ minors of $\nabla\bb{\Delta^1\Delta^1\v F(0;0)}$ is nonzero; again we will look more closely at these in the next step. 
%This tells us that all of the $m^{1}_{3,j}$ vanish for $j=n^2+1,...,n^2+n+2$, so at least one of the remaining minors, $\op B_{3,jk}$ among them, must be non-zero. 
%Since $\op B_{3,jk}(0;0)\neq0$ for some $j,k$, the gradient $\nabla\op B_{2,j}(0;0)$ lies outside of $\op T$, so $\op T_{\op B_{2,j}}^k$ is $n$-dimensional for that $j,k$. 

Assume instead that $\tb_1=\tb_2=\tb_3=1$. Now $\nabla\v F(0;0)$, $\nabla\Delta^1\v F(0;0)$, and $\nabla\Delta^1\Delta^1\v F(0;0)$ all have corank 1, so all of the $n\times n$ minors of $\nabla\Delta^1\Delta^1\v F(0;0)$ are zero. 
Recalling 
$$\nabla(\Delta^{1}\Delta^{1}\v F)=\bb{\nabla\v F,\nabla \op B_1,\nabla\op B_{2,1},...,\nabla\op B_{2,n},\nabla\op B_1}\;,$$ these minors number $\sfrac{(2(n+1))!}{n!(n+2)!}$ (the binomial coefficient for choosing $n$ rows from $2(n+1)$). 
One of these is $\op B_1=\det\nabla\v F$, another $2n$ of them are the functions $\op B_{2,1},...,\op B_{2,n+1}$ repeated twice (because $\op B_1$ is repeated twice in $\nabla\Delta^1\Delta^1\v F$). 
Another $n^2$ are the determinants of Jacobian matrices in which some $k_2^{th}$ row of $\nabla\v F$ is swapped for one row $\nabla m^{1}_{2,k_1}$, with $k_1,k_2=1,...,n$, and these are precisely the functions $\op B_{3,k_1k_2}$, for $k_1,k_2=1,...,n$. Let
\begin{subequations}\label{minors3}
\begin{align}
m^{1}_{3,l(k_1,k_2)}&=\op B_{3,k_1k_2}(0;0)\quad{\rm for} \quad k_1,k_2=1,...,n,\nonumber \\
				&\qquad\qquad\;\;{\rm with}\quad l(k_1,k_2)=k_1+n(k_2-1)\;,\\
m^{1}_{3,j}&=\op B_{2,j}(0;0)\quad\;{\rm for} \;\;\;\;\quad j=n^2+1,...,n^2+n\;, \\
m^{1}_{3,j}&=\op B_{2,j}(0;0)\quad\;{\rm for} \;\;\;\;\quad j=n^2+n+1,...,n^2+2n\;, \\
m^{1}_{3,j}&=\op B_1(0;0)\quad\;\;\;{\rm for} \;\;\;\;\quad j=n^2+2n+1,...,n^2+2n+2\;, \\
m^{1}_{3,j}&=0	\qquad\quad\;\;{\rm for} \;\qquad j=n^2+n+3,...,\sfrac{(2(n+1))!}{n!(n+2)!}\;.
\end{align} 
\end{subequations}
We have claimed in the last line that the remaining $m^{1}_{3,j}$ for $j=n^2+n+3,...,\sfrac{(2(n+1))!}{n!(n+2)!}$ vanish. These are the determinants of Jacobian matrices formed from $0\le d\le n-2$ rows of $\nabla\v F$ and $2\le d'\le n$ rows from $\bb{\nabla \op B_1,\nabla\op B_{2,1},...,\nabla\op B_{2,n},\nabla\op B_1}$, but by \cref{thm:minors} this has rank $n-1$, so all its $n\times n$ minors vanish. 
So $\op B_{3,k_1k_2}(0;0)=0$ for all $k_1,k_2=1,...,n$ (as well as $\op B_1(0;0)=\op B_{2,k_1}(0;0)=0$ for all $k_1=1,...,n$). 
Hence if $\tb_1=\tb_2=\tb_3=1$ then $\op B_1(0;0)=\op B_{2,k_1}(0;0)=\op B_{3,k_1k_2}(0;0)=0$ for all $k_1,k_2=1,...,n$. Conversely, if $\op B_1(0;0)=\op B_2(0;0)=\op B_3(0;0)=0$, then by \cref{thm:Brall}, $\op B_{2,k_1}(0;0)=\op B_{3,k_1k_2}(0;0)=0$ for all $k_1,k_2=1,...,n$, implying $\tb_1=\tb_2=\tb_3=1$. 

Now if $\tb_4=0$ the singularity is a {\it swallowtail}, and the $\bb{2(n+1)+\chi}\times n$ matrix $$\nabla\Delta^1\Delta^1\Delta^1\v F=\cc{\nabla\v F,\nabla\op B_1,\nabla m^{1}_{2,1},...,m^{1}_{2,n+1},\nabla m^{1}_{3,1},...,\nabla m^{1}_{3,\chi}}$$ where $\chi=\sfrac{(2(n+1))!}{n!(n+2)!}$, must have rank $n$ since $\tb_3=0$, so at least one of its $n\times n$ minors must be nonzero; as usual, inspection of these is left to the next step. 

And so on. At each successive symbol of length $r$, the vanishing of minors is equivalent to the vanishing of all of the quantities $\v F(0;0)$, $\op B_1(0;0)$, ... $\op B_{r,k_1...k_{r-1}}$. At the next order, assuming $\tb_1=\tb_2=\tb_3=\tb_4=1$, for example, the minors include the functions $\op B_{4,k_1k_2k_3}$, and we can define
%\begin{subequations}
\begin{align}\label{minors4}
m^{1}_{4,l(k_1,k_2,k_3)}&=\op B_{4,k_1k_2k_3}\quad{\rm for} \quad k_1,k_2,k_3=1,...,n, \nonumber\\
			&\quad{\rm where}\quad l(k_1,k_2,k_3)=k_1+n(k_2-1)+n^2(k_3-1)\;,%\\
%m^{1}_{4,l(j,k)}&=\op B_{3,jk}\quad{\rm for} \quad j,k=1,...,n,  \nonumber\\
%			&\qquad\qquad\quad\;\;  l(j,k)=j+n(k-1)+n^3\;,\\  
%m^{1}_{4,j}&=\op B_{2,j}\quad\;{\rm for} \;\;\;\;\quad j=n^3+n^2+1,...,n^3+n^2+n\;, \\
% m^{1}_{4,j}&=\op B_1\quad\;\;\;{\rm for} \;\;\;\;\quad j=n^3+n^2+n+1,...,n^3+n^2+n+2\;. %\\
\end{align} 
%\end{subequations}
while the remaining $m^{1}_{4,j}$ for $j=n^3+1,...,\sfrac{\bb{2(n+1)+\chi}!}{n!(\bb{2(n+1)+\chi}-n)!}$, consist of the functions $\op B_1$, $\op B_{2,k_1}$, $\op B_{3,k_1k_2}$, for $k_1,k_2=1,...,n$, as well as determinants of Jacobian matrices formed from $0\le d\le n-2$ rows of $\nabla\v F$ and $2\le d'\le n$ rows from $\bb{\nabla \op B_1,\nabla m^{1}_{2,1},...,\nabla m^{1}_{2,n+1},\nabla m^{1}_{3,1},...,\nabla m^{1}_{3,\chi}}$, whose minors vanish by \cref{thm:minorsall}. At each successive codimension, applying \cref{thm:Brall} shows equivalence between the symbols $\tb_1=...=\tb_r=0$ and the conditions $\op B_1(0;0)=...=\op B_r(0;0)=0$.

\rightline{$\square$}

\bigskip

This does not mean that the `underlying catastrophes' defined in \cite{j22cat} are equivalent to the Thom-Boardman singularities, indeed they are not. But they are consistent in that a codimension $r$ underlying catastrophe (where $\v F(0;0)=\op B_1(0;0)=...=\op B_r(0;0)=0$) is a codimension $r$ {\it Morin} singularity, i.e. the conditions are {\it sufficient} to define a zero of $\v F$ lying at such a singularity of $\v F$ viewed as a mapping. 

Importantly in \cref{def:sings} we place a further restriction of `fullness', which ultimately is not {\it necessary} to define a singularity in the sense of $\v F$ as a mapping, but ensures that $\v F$ is determined such that the conditions $\v F(0;0)=\op B_1(0;0)=...=\op B_r(0;0)=0$ are solveable. We complete this section by showing how this follows from the $\op G_{r,K(r-1)}$ conditions.

\bigskip
%%%%%%%%%%%%%%%%%%%%%%%%%%%%%%%%%%%%%%%%%%%%%%%%%%
\subsection{The $\op G$ conditions for being `full' }\label{sec:Gfull}

%As shown above, the $\op B$ conditions are necessary and sufficient to locate singularities {\it in the sense of a mapping $\v F:\mathbb R^n\to\mathbb R^n$}. The full \BG~conditions are used to find a catastrophe in the space of the variables and parameters, that is a singularity of $\v F:\mathbb R^n\times\mathbb R^r\to\mathbb R^n\times\mathbb R^r$ writing
%\begin{align}
% F^{(p)}(x;\alpha_1,...,\alpha_p)=\bb{\v F(x;\alpha_1,...,\alpha_p),\alpha_1,...,\alpha_p}\;.
%\end{align}
%Just as $\nabla=\partial/\partial x$ above, let $\square_p=\bb{\sfrac{\partial\;}{\partial(\v x,\alpha_1,...,\alpha_p)}}$. In the trivial case of $p=0$ we have $F^{0}=F$ and $\square_0=\nabla$. 

The only place we have needed to make reference to the parameter dependence of $\v F$ above is in the proof of \cref{thm:Buni}. Taking $\v F:\mathbb R^n\times\mathbb R^r\to\mathbb R^n$, denote again the gradient operator as $\square=(\nabla_{\v x},\nabla_\alpha)$. 
The conditions $\op G_{r,K(r-1)}(0;0)\neq0$ simply state that the extended Jacobian 
\begin{align}\label{Gone}
\op G_r=\frac{\partial(\v F,\op B_1,...,\op B_r)}{\partial(\v x,\alpha_1,...,\alpha_r)}
\end{align}
has full rank $n+r$ at $(\v x,\alphab)=(0,0)$, as do its variants $\op G_{r,k_1...k_{r-1}}$ for the alternative permutations of the $\op B_{i,k_1...k_{r-1}}$. At this point, the gradient vectors $\square F_1$, ..., $\square F_n$, $\square\op B_1$, ..., $\square\op B_r$, span an $(n+r)$-dimensional linear subspace $\op T_{\v x,\alpha}\subset\mathbb R^n\times\mathbb R^r$, which implies by the implicit function theorem that the problem $\v F=\op B_1=...=\op B_r=0$ is solvable at isolated points $(\v x,\alphab)=(\v x_*,\alpha_{1*},...,\alpha_{r*})$.

We can illustrate this using the primary form of a codimension $r$ singularity given in \cref{primary}, namely
\begin{align*}
\v F=&\Big(\;f(x_1,\alphab)+\underline{\tau}\cdot\underline{x}\;,\;\lambda_2x_2\;,\;...\;,\;\lambda_nx_n\;\Big)\\
{\rm where}&\qquad
f(x_1)=x_1^{r+1}+\sum_{i=1}^r\alpha_i x_1^{i-1}\;,\nonumber
\end{align*}
with $\underline{\tau}\cdot\underline{x}=\tau_2x_2+...+\tau_nx_n$, and where the $\tau_i$ and $\lambda_i$ are non-zero constants. 

Calculating $\op B_r$ from \cref{Br}, and the function $\op G_r$ from \cref{Gone}, we have
\begin{align}\label{primaryBG}
\op B_r&=(\lambda_2...\lambda_n)^rf^{(r)}(x_1)\;,\nonumber\\
\op G_r&=c_1(\lambda_2...\lambda_n)^{r+1}f^{(r+1)}(x_1)\;,
\end{align}
and calculating $\op B_{r,k_1...k_{r-1}}$ from \cref{Bri} and $\op G_{r,k_1...k_{r-1}}$ from \cref{Gri}, we have
\begin{align}\label{primaryBG}
\op B_{r,k_1...k_{r-1}}&=c_2\lambda_2^{p_2}\lambda_n^{p_n}\tau_2^{q_2}...\tau_n^{q_n}\op B_r\;,\nonumber\\
\op G_{r,k_1...k_{r-1}}&=c_3\lambda_2^{s_2}...\lambda_n^{s_n}\tau_2^{t_2}...\tau_n^{t_n}\op G_r\;.
\end{align}
The actual values of the various numerical constants $c_1,c_2,c_3$ and $-r\le p_i,q_i\le0$ and $-r-1\le s_i,t_i\le0$ are not important, and depend quite complicatedly on the values of $r$ and the $k_i$s. The point is that we need all $\lambda_i$ and $\tau_i$ to be non-vanishing for the $\op G$ determinants to be non-vanishing, and thereby ensure that the $\op B$ determinants are non-trivial and have well-defined roots coinciding with those of the functions $f^{(r)}$. 
In a following paper \cite{jc23tbprocedure} we intend to prove that the {\it underlying catastrophes} of \cref{def:sings} are in fact transformable, as mappings though not as vector fields, to the functions \cref{primary}.

%Examples of the singularities that lie beyond the \BG~conditions because they are not {\it full} were given in \cite{j22cat}, showing that some scenarios become detectable once redundant dimensions are discarded, which is related to the property in \cite{j22cat} of being {\it full}, meaning that the \BG~conditions are solveable, as a weaker but more practical alternative to being {\it generic}. 

%%%%%%%%%%%%%%%%%%%%%%%%%%%%%%%%%%%%%%%%%%%%%%%%%%%%%%
%%%%%%%%%%%%%%%%%%%%%%%%%%%%%%%%%%%%%%%%%%%%%%%%%%%%%%%
%\section{}%\label{sec:}
%
%\begin{align}
%\end{align}
%
%\begin{align}
%\end{align}
%
%\begin{align}
%\end{align}
%
%\begin{align}
%\end{align}
%
%\begin{align}
%\end{align}
%
%\begin{align}
%\end{align}
%
%\begin{align}
%\end{align}

%%%%%%%%%%%%%%%%%%%%%%%%%%%%%%%%%%%%%%%%%%%%%%%%%%%%%%
%%%%%%%%%%%%%%%%%%%%%%%%%%%%%%%%%%%%%%%%%%%%%%%%%%%%%%
\section{Closing remarks}\label{sec:conc}

The results above justify the methodology laid out in \cite{j22cat} and connect it with the theory of Thom-Boardman singularities. The hope is that these ideas make the application of Thom's ideas to vector fields and spatiotemporal systems much more practical, by providing new ways to find and characterize high codimension bifurcation points.

 A number of examples were given in \cite{j22cat} to illustrate the importance of being {\it full}. 
A non-trivial example illustrating the importance of satisfying \cref{subrankfull}, even when \cref{singdeg} holds, is given by $\v F=(x_1+x_2^2,x_1^2+\alpha_1 x_1+\alpha_2+x_2^2+kx_1)$. If $k_1=0$ then $\subrank\v F(0;0)=0$ and we find $\op B_1(0;0)=\op B_{21}(0;0)=0\neq\op B_{22}(0;0)$, violating \cref{thm:Br}, while the fullness conditions $\op G_{2,1}(0;0)\neq0$ and $\op G_{2,2}(0;0)\neq0$ are still satisfied. If $k_1\neq0$ then \cref{subrankfull} holds. One then finds that there is a fold at the origin as $\alpha_2$ passes through zero, while the event that unfolds with $\alpha_1$ is degenerate (not full), and there is no higher codimension catastrophe since $\op B_{21}(0;0),\op B_{22}(0;0)\neq0$.

We have claimed that any bifurcation has an underlying catastrophe, and established the validity of the conjecture with a suitable definition that relates to zeros of vector field encountering certain corank 1 singularities. 
As already noted in \cite{j22cat}, underlying catastrophes are not equivalence classes, and each one may represent several classes of bifurcations with different stability properties. The idea is that the concept of an underlying catastrophe is used to {\it find} a bifurcation point in a vector field, after which standard bifurcation analysis can be carried out to properly classify it, if required. 

Here we have described more the relation of underlying catastrophes to singularities (rather than bifurcations), and this proves to be rather different, in that there appear to be singularities within Thom's classification that cannot be identified with any underlying catastrophe, because they require too many conditions to identify them. This perhaps rules out some singularities as unfindable by the solution of implicit conditions, but what singularities this applies to, and what this means for their role in applications, remains to be studied.

At present this analysis only applies to corank 1 cases. For these, the \BG~conditions reduce the number of minors needed to calculate the Boardman symbols, which increases factorially with the codimension $r$, to just $\op B_1,...,\op B_r$.
For what other singularities or bifurcations is this possible? It seems obvious, for example, that corank 2 singularities should be related to {\it umbilic} catastrophes, and the first steps towards extending the \BG~conditions for these has been taken in \cite{j24cat}.

\bigskip\bigskip

\noindent{\Large\bf Appendix}

%\newpage
\appendix
%%%%%%%%%%%%%%%%%%%%%%%%%%%%%%%%%%%%%%%%%%%%%%%%%%%%%%
%%%%%%%%%%%%%%%%%%%%%%%%%%%%%%%%%%%%%%%%%%%%%%%%%%%%%%
\section{Counting minors}\label{sec:count}

The number of minors involved in calculating the Thom-Boardman symbols in \cref{sec:TB} grows `superfactorially', that is, faster than factorially in the codimension $r$. To see this let us calculate the number, $N_j$, of new minors that must be calculated to find each successive $j^{th}$ Boardman symbol, following exactly the procedure in \cref{sec:TB}. 

To find the first Boardman symbol $\tb_1$, one must calculate the $(n-{i_1}+1)\times(n-{i_1}+1)$ minors of $\nabla\v F=(\nabla f_1,...,\nabla f_n)$, which we call $m^{i_1}_1,...,m^{i_1}_{N_1}$. Each minor is formed by choosing $(n-{i_1}+1)$ rows and columns from an $n\times n$ matrix, and so
\begin{align*}
		N_1&=\Cr{n}{n-i_1+1}\cdot\Cr{n}{n-i_1+1}
\end{align*}
%\begin{align*}
%&\Big(\begin{array}{c}n\\m\end{array}\Big)=\sfrac{n!}{m!(n-m)!}\\
%		N_1&=\Cr{n}{n-i_1+1}\cdot\Cr{n}{n-i_1+1}\\
%			&=\sfrac{n!}{(n-i_1+1)!(i_1-1)!}\sfrac{n!}{(n-i_1+1)!(i_1-1)!}
%\end{align*}
where $\Cr{n}{m}$ denotes the binomial coefficient. In finding the second Boardman symbol $\tb_2$, one then calculates the $(n-{i_2}+1)\times(n-{i_2}+1)$ minors of $\nabla\Delta^{i_1}\v F=(\nabla f_1,...,\nabla f_n,m^{i_1}_1,...,m^{i_1}_{N_1})$, which we call $m^{i_2}_1,...,m^{i_2}_{N_2}$, where each minor is formed by choosing $(n-{i_2}+1)$ rows and columns from an $n\times(n+N_1)$ matrix, and so
\begin{align}
		N_2&=\Cr{n}{n-i_2+1}\cdot\Cr{n+N_1}{n-i_2+1}\;.
\end{align}
Continuing iteratively is easy to see that
\begin{align*}
		N_3&=\Cr{n}{n-i_3+1}\cdot\Cr{n+N_1+N_2}{n-i_3+1}
\end{align*}
and so on. 

%\begin{align}
%\nabla\v F&=(\nabla f_1,...,\nabla f_n)\;\sim\; n\times n\nonumber\\\Rightarrow\qquad 
%		N_1&=\Cr{n}{n-i_1+1}\cdot\Cr{n}{n-i_1+1}\nonumber\\\\
%\nabla\Delta^{i_1}\v F&=(\nabla f_1,...,\nabla f_n,m^{i_1}_1,...,m^{i_1}_{N_1})\;\sim\; n\times(n+N_1)\nonumber\\\Rightarrow\qquad 
%		N_2&=\Cr{n}{n-i_2+1}\cdot\Cr{n+N_1}{n-i_2+1}\nonumber\\\nonumber\\ 
%\nabla\Delta^{i_1}\v F&=(\nabla f_1,...,\nabla f_n,m^{i_1}_1,...,m^{i_1}_{N_1},m^{i_2}_1,...,m^{i_2}_{N_2})\;\sim\; n\times(n+N_1+N_2)\nonumber\\\nonumber\Rightarrow\qquad 
%		N_3&=\Cr{n}{n-i_3+1}\cdot\Cr{n+N_1+N_2}{n-i_3+1}\nonumber\\\nonumber\vdots\;\;&=\;\;\;\vdots\nonumber\\\Rightarrow\qquad 
%		N_j&=\Cr{n}{n-i_j+1}\cdot\Cr{\sum_{k=0}^{j-1}N_k}{n-i_j+1}\;,\quad N_0=n\;.
%\end{align}
%

That is, to calculate the Boardman symbol for a codimension $r$ singularity we must evaluate $\op N$ minors, the number
\begin{align}
		\op N(r)&=\sum_{j=1}^rN_j
%		\qquad{\rm where}\\
%		N_j&=\Cr{n}{n-i_j+1}\cdot\Cr{\sum_{k=0}^{j-1}N_k}{n-i_j+1}\;,\qquad N_0=n\;,
\end{align}
where 
\begin{align}
		N_j&=\Cr{n}{n-i_j+1}\cdot\Cr{\sum_{k=0}^{j-1}N_k}{n-i_j+1}\;,\quad N_0=n\;.
\end{align}
and the brackets denote the binomial coefficient
\begin{align}
\Big(\mbox{\footnotesize$\begin{array}{c}n\\m\end{array}$}\Big)&=\sfrac{n!}{m!(n-m)!}\;,
\end{align}
\Cref{tab:count} shows $\op N(r)$ for the first few codimensions $r$ for corank 1 (when all $i_j=1$), in systems of different dimension $n$, and the number of $\op B$ determinants this reduces to for an underlying catastrophe.

\begin{table}[h!]
\raggedright\qquad\qquad\qquad
Thom-Boardman singularities\\\qquad\qquad\qquad
\begin{tabular}{|r|l|l|l|l|l|}\hline
codim. $r=\;\rightarrow$	&$1$	&$2$	&$3$	&$4$&$5$	\\
 dim. $n=\;\downarrow$		&&&&&\\\hline
$1$	&2	&4	&8	&16 &32	\\
$2$	&3	&6	&21	&231	 &26796\\
$3$	&4	&8	&64	& 41728 & $12\times10^{12}$	\\%\footnotesize12 108 775 752 704\\
$4$	&5	&10	&220	&\footnotesize94967015 &$3\times10^{30}$ 	\\\hline%\footnotesize3 389 064 862 738 503 046 956 274 404 130 \\%&	\\\hline
\end{tabular}
%\quad$\Rightarrow$\quad
\\\medskip\qquad\qquad\qquad
%\quad$\begin{array}{c}\Rightarrow\\\\\Rightarrow\\\\\Rightarrow\end{array}$\quad
underlying catastrophes\\\qquad\qquad\qquad
\begin{tabular}{|r|l|l|l|l|l|}
\hline
codim. $r=\;\rightarrow$	&$1$	&$2\hspace{0.2cm}$	&$3\hspace{0.4cm}$&$4\hspace{1.1cm}$&$5\hspace{1.35cm}$	\\
 dim. $n=\;\downarrow$		&&&&&\\\hline
$1$	&2	&3	&4	&5	&6	\\
$2$	&3	&4	&5	&6	&7	\\
$3$	&4	&5	&6	&7	&8	\\
$4$	&5	&6	&7	&8	&9	\\\hline
\end{tabular}
\caption{\small\sf The number problem: these are the number of minors that are calculated to identify a singularity in the standard Thom-Boardman classification (top), versus the number of $\op B$ determinants needed to identify an underlying catastrophe (bottom), up to dimension $n=4$ and codimension $r=5$.
}
\label{tab:count}
\end{table}

{\small
\bibliography{BG-TB-unifiedR.bbl}
\bibliographystyle{plain} 
}

\end{document}